\documentclass[a4paper,UKenglish,cleveref, autoref, thm-restate]{lipics-v2021}

\usepackage[rgb]{xcolor}
\usepackage{environ}
\usepackage{lipsum}
\usepackage{amsfonts}
\usepackage{cite}
\usepackage{xspace}
\usepackage{array}
\usepackage{xparse}
\usepackage{amsmath}
\usepackage{complexity}
\usepackage{url}[spaces, hyphens]
\usepackage{ulem}
\usepackage[scr=kp]{mathalfa} % remove to restore \mathscr font
\usepackage{outlines}
\usepackage{svg}
\usepackage{tabularx}
\usepackage{booktabs}
\usepackage{ltablex}
\usepackage{makecell}

\usepackage{dsfont}
\usepackage{tikz}

\usetikzlibrary{mindmap,positioning}
\usetikzlibrary{graphs}
\usetikzlibrary[graphs]
\usetikzlibrary{patterns}
\usetikzlibrary{arrows,automata,positioning}
\usetikzlibrary{decorations.pathreplacing}
\usetikzlibrary{decorations}
\usetikzlibrary{decorations.pathmorphing}
\usetikzlibrary{shapes.geometric}
\usetikzlibrary{quotes}
\usetikzlibrary{shapes.callouts}
\usetikzlibrary{arrows.meta}
\usetikzlibrary{calc}
\usetikzlibrary{fit}

\definecolor{goodred}{HTML}{CC6677}
\definecolor{goodblue}{HTML}{332288}
\definecolor{goodyellow}{HTML}{DDCC77}
\definecolor{goodgreen}{HTML}{117733}
\definecolor{goodcyan}{HTML}{88CCEE}
\definecolor{goodwine}{HTML}{882255}
\definecolor{goodteal}{HTML}{44AA99}
\definecolor{goodolive}{HTML}{999933}
\definecolor{goodpurple}{HTML}{AA4499}

\tikzset{faded/.style={gray,very thin}}
\tikzset{vertex/.style={draw,circle,minimum size=5pt,inner sep=0pt}}
\tikzset{novertex/.style={circle,minimum size=5pt,inner sep=0pt}}
\tikzset{blackvertex/.style={draw,circle,minimum size=5pt,inner sep=0pt, fill=black}}
\tikzset{redvertex/.style={draw,circle,minimum size=5pt,inner sep=0pt, fill=goodred}}
\tikzset{redvertexfaded/.style={draw,circle,faded,minimum size=5pt,inner sep=0pt, fill=goodred!50}}
\tikzset{greenvertex/.style={draw,circle,minimum size=5pt,inner sep=0pt, fill=green}}
\tikzset{greenvertexfaded/.style={draw,circle,faded,minimum size=5pt,inner sep=0pt, fill=green!50}}
\tikzset{bluevertex/.style={draw,circle,minimum size=5pt,inner sep=0pt, fill=goodcyan}}
\tikzset{bluevertexfaded/.style={draw,circle,faded,minimum size=5pt,inner sep=0pt, fill=goodblue!50}}
\tikzset{yellowvertex/.style={draw,circle,minimum size=5pt,inner sep=0pt, fill=yellow}}
\tikzset{yellowvertexfaded/.style={draw,circle,faded,minimum size=5pt,inner sep=0pt,
fill=yellow!50}}
\tikzset{arrow/.style={-{Latex[scale=1]},shorten >= 0pt}}
\tikzset{edge/.style={-{Latex[scale=1]},shorten >= 0pt}}

\tikzset{snake it/.style={decorate, decoration=snake}}
\tikzset{big snake/.style={decorate, decoration={snake,segment length=7mm, amplitude=3mm}}}
\tikzset{ultra snake/.style={decorate, decoration={snake,segment length=7mm, amplitude=6mm}}}

\newcommand{\lipItem}[1]{\textcolor{lipicsGray}{\sffamily\bfseries\upshape\mathversion{bold}#1}}
\pdfoutput=1 %uncomment to ensure pdflatex processing (mandatatory e.g. to submit to arXiv)
\hideLIPIcs  %uncomment to remove references to LIPIcs series (logo, DOI, ...), e.g. when
\title{Finding subdigraphs in digraphs of bounded directed treewidth} %TODO Please ad

\titlerunning{Finding subdigraphs in digraphs of bounded directed treewidth}

\author{Raul Lopes}{LIRMM, Université de Montpellier, CNRS, Montpellier,
France\\
Hamburg University of Technology, Institute for Algorithms and Complexity, Hamburg, Germany}{rtlopes@protonmail.com}{https://orcid.org/0000-0002-7487-3475}{French project ELIT (ANR-20-CE48-0008-01) and HIDSS-0002 DASHH (Data Science in Hamburg - Helmholtz Graduate School for the Structure of Matter).}

\author{Ignasi Sau}{LIRMM, Université de Montpellier, CNRS, Montpellier,
France}{ignasi.sau@lirmm.fr}{https://orcid.org/0000-0002-8981-9287}{French project ELIT (ANR-20-CE48-0008-01).}

\authorrunning{ }

\Copyright{ }

 %TODO mandatory; please add
\category{} %optional, e.g. invited paper

\relatedversion{} %optional, e.g. full version hosted on arXiv, HAL, or other respository/website
\nolinenumbers %uncomment to disable line numbering

\EventEditors{Inge Li Gørtz,
  Simon J. Puglisi, and
Martin Farach-Colton,}
\EventNoEds{2}
\EventLongTitle{ESA 2023?}
\EventShortTitle{ESA 2023?}
\EventAcronym{ESA}
\EventYear{2023}
\EventDate{September 4--6, 2023}
\EventLocation{Amsterdam, The Netherlands}
\EventLogo{}
\SeriesVolume{XXX}
\ArticleNo{ }
\DeclareMathOperator{\order}{\text{\sf ord}\xspace}
\DeclareMathOperator{\dtw}{\text{\sf dtw}\xspace}

\DeclareMathOperator{\breakability}{\text{\sf b}\xspace}
\DeclareMathOperator{\Ocal}{\mathcal{O}\xspace}

\newcommand{\width}{{\sf width}\xspace}
\newcommand{\pname}{Rooted Stars-Paths Subdigraph Isomorphism\xspace}
\newcommand{\pshort}{RSPSI\xspace}
\newcommand{\pddp}{Subset Avoiding Directed Disjoint Paths\xspace}
\newcommand{\pddpshort}{SADDP\xspace}
\newcommand{\slack}{{\sf slack}\xspace}
\newcommand{\xarray}{{\sf\bf x}\xspace}
\newcommand{\sarray}{{\sf\bf s}\xspace}
\newcommand{\parray}{{\sf\bf p}\xspace}

\newtheorem{condition}[theorem]{Condition}
\Crefname{condition}{Condition}{Conditions}

\renewcommand{\mid}{\bigm|}

\makeatletter

\newcolumntype{\expand}{}
\long\@namedef{NC@rewrite@\string\expand}{\expandafter\NC@find}

\NewEnviron{boxproblem}[2][]{%
  \def\boxproblem@arg{#1}%
  \def\boxproblem@framed{framed}%
  \def\boxproblem@lined{lined}%
  \def\boxproblem@doublelined{doublelined}%
  \ifx\boxproblem@arg\@empty%
  \def\boxproblem@hline{}%
  \else%
  \ifx\boxproblem@arg\boxproblem@doublelined%
  \def\boxproblem@hline{\hline\hline}%
  \else%
  \def\boxproblem@hline{\hline}%
  \fi%
  \fi%
  \ifx\boxproblem@arg\boxproblem@framed%
  \def\boxproblem@tablelayout{|>{\bfseries}lX|c}%
  \def\boxproblem@title{\multicolumn{2}{|l|}{%
      \raisebox{-\fboxsep}{\textsc{\normalsize #2}}%
  }}%
  \else
  \def\boxproblem@tablelayout{>{\bfseries}lXc}%
  \def\boxproblem@title{\multicolumn{2}{l}{%
      \raisebox{-\fboxsep}{\textsc{\normalsize #2}}%
  }}%
  \fi%
  \smallskip\par\noindent%
  \begin{minipage}{\textwidth}
  \begin{tabularx}{\textwidth}{\expand\boxproblem@tablelayout}%
    \boxproblem@hline%
    \boxproblem@title\\[2\fboxsep]%
    \BODY\\\boxproblem@hline%
  \end{tabularx}%
  \end{minipage}
  \medskip\par%
}
\makeatother

\definecolor{mid-green}{rgb}{0.15,0.65,0.15}
\definecolor{dark-green}{rgb}{0.15,0.25,0.15}
\definecolor{dark-red}{rgb}{0.7,0.15,0.15}
\definecolor{dark-blue}{rgb}{0.15,0.15,0.9}
\definecolor{medium-blue}{rgb}{0,0,0.5}
\definecolor{gray}{rgb}{0.5,0.5,0.5}
\definecolor{color-Ig}{rgb}{0.15,0.7,0.15}
\definecolor{darkmagenta}{rgb}{0.30, 0.0, 0.30}

\hypersetup{
  colorlinks, linkcolor={dark-blue},
citecolor={mid-green}, urlcolor={dark-blue}}

\renewcommand{\NP}{{\sf NP}\xspace}
\renewcommand{\P}{{\sf P}\xspace}
\renewcommand{\FPT}{{\sf FPT}\xspace}
\renewcommand{\XP}{{\sf XP}\xspace}

\begin{document}

\maketitle

\begin{abstract}
It is well known that directed treewidth does not enjoy the nice algorithmic properties of its undirected counterpart. There exist, however, some positive results that, essentially, present \XP algorithms for the problem of finding, in a given digraph $D$, a subdigraph isomorphic to a digraph $H$ that can be formed by the union of $k$ directed paths (with some extra properties), parameterized by $k$ and the directed treewidth of $D$. Our motivation is to tackle the following question: Are there subdigraphs, other than the directed paths, that can be found efficiently in digraphs of bounded directed treewidth?  In a nutshell, the main message of this article is that, other than the directed paths, the only digraphs that seem to behave well with respect to directed treewidth are the stars. For this, we present a number of positive and negative results, generalizing several results in the literature, as well as some directions for further research. 
\end{abstract}

%\red{\fbox{ESA 2023 guidelines: 12 pages, without taking into account titlepage, nor references.}}

%\newpage
%\setcounter{page}{1}

\section{Introduction}
\label{sec:intro}

Treewidth of undirected graphs~\cite{RobertsonS90a} is probably the most successful graph parameter from an algorithmic viewpoint, as capitalized by the celebrated Courcelle's theorem~\cite{COURCELLE199012}, stating that all problems that are expressible in (the very powerful) monadic second-order logic can be solved in linear time when restricted to graphs of bounded treewidth; using terminology from parameterized complexity, they are \textit{fixed-parameter tractable} (\FPT). Unfortunately, the directed counterpart of treewidth, called \textit{directed treewidth}~\cite{Reed99, Johnson.Robertson.Seymour.Thomas.01}, does not enjoy the nice algorithmic properties of undirected treewidth. Indeed, on the one hand, directed treewidth seems harder to compute (or to approximate) than its undirected counterpart~\cite{Bodlaender96,KorhonenL23,doi:10.1137/21M1452664,Johnson.Robertson.Seymour.Thomas.01}. On the other hand, and closer to the topic of this article, some ``basic'' problems that fit the setting of Courcelle's theorem~\cite{COURCELLE199012} cannot be solved efficiently in digraphs when parameterized by directed treewidth, as observed by Lampis et al.~\cite{LampisKM11}. This is the case, for instance, of {\sc Directed Hamiltonian Path} that is ${\sf W}[2]$-hard parameterized by directed treewidth (hence, unlikely to be \FPT), or {\sc Max Directed Cut} that is even \NP-hard restricted to directed acyclic graphs (DAGs), which have directed treewidth zero~\cite{LampisKM11}. Ganian et al.~\cite{GanianHK0ORS16} provide concrete reasons for which digraph width measures (including directed treewidth) cannot have the same nice algorithmic properties as undirected treewidth. Since our focus is algorithmic, in this discussion we deliberately omit the recent impressive progress on {\sl structural} properties of directed treewidth~\cite{HatzelKMM24,GiannopoulouKKK20,KawarabayashiK15}. 

However, there are a few positive algorithmic results concerning directed treewidth, two of which are particularly relevant to us.
Namely, in the article where directed treewidth was introduced, Johnson et al.~\cite{Johnson.Robertson.Seymour.Thomas.01} proved that {\sc Directed Hamiltonian Path} and {\sc Directed Disjoint Paths} can be solved in \XP time (that is, in polynomial time for each fixed value of the parameter; see~\cite{DF13,CyganFKLMPPS15} for basic background on parameterized complexity) parameterized by directed treewidth (plus the number of paths in the latter problem). Later, de Oliveira Oliveira~\cite{Oliveira16} provided an algorithmic meta-theorem for directed treewidth generalizing the previous results, by showing, informally and in a simplified statement, that the problem of deciding whether a digraph $D$ contains a subdigraph $H$ consisting of the union of $k$ directed paths and satisfying a given monadic second-order formula $\varphi$ can be solved in \XP time parameterized by $k$, the size of $\varphi$, and the directed treewidth of $D$. 

It is worth noting that both algorithms in~\cite{Johnson.Robertson.Seymour.Thomas.01,Oliveira16} have an \XP dependence on the directed treewidth, and, more crucially,  that these \XP algorithms are for problems consisting in finding {\sl directed paths} satisfying certain properties, namely being pairwise disjoint and rooted at prescribed terminals in~\cite{Johnson.Robertson.Seymour.Thomas.01}, or satisfying a monadic second-order formula in~\cite{Oliveira16}. 

The question that arises naturally from the previous discussion, and which is our main motivation, is the following: are there subdigraphs, other than the directed paths, that can be found efficiently in digraphs of bounded directed treewidth? 

\medskip
\noindent\textbf{Our contribution.} In a nutshell, the main message of this article is that, other than the (directed) paths, the only digraphs that seem to behave well with respect to directed treewidth are the {\sl stars}. In what follows, we make the former vague statement more precise. 

Driven by the above question, we are interested in the following problem, called {\sc Subdigraph Isomorphism}: given a host digraph $D$ of directed treewidth at most $w$ and a digraph $H$, decide whether $D$ contains a subdigraph isomorphic to $H$.\footnote{A clarification is in place here. Note that the way we have defined {\sc Subdigraph Isomorphism}, it does {\sl not} encompass, for instance, the {\sc Directed Disjoint Paths} problem, where the input contains pairs of vertices, called the \textit{terminals}, to be joined by the corresponding paths. Nevertheless, as it will become clear, the \XP algorithm that we present solves the {\sl rooted} version (that is, with terminals) as well.} To the best of our knowledge {\sc Subdigraph Isomorphism} in its full generality, and considering as a parameter the directed treewidth of the input digraph, has been unexplored in the literature, in contrast with its undirected counterpart (see for instance~\cite{MarxP14}).

Unsurprisingly, {\sc Subdigraph Isomorphism}\ is \NP-complete in general (for instance, if $H$ is a transitive tournament, there is a trivial reduction from {\sc Clique} in undirected graphs), so our focus is on parameterized algorithms when taking as the parameters (combinations of) $w$ and some parameter $\kappa$ depending on $H$. Observe that if we take $\kappa(H)=|V(H)|$, then the problem can be trivially solved in time $|V(D)|^{\Ocal(\kappa(H))}$ by just brute-forcing over all vertex sets of $D$ of size $|V(H)|$ and checking whether the subdigraph of $D$ induced by them contains a digraph isomorphic to $H$. 
Thus, in order to face non-trivial questions, we consider parameters $\kappa(H)$ smaller than the size of $H$. In the spirit of the meta-algorithm of de Oliveira Oliveira~\cite{Oliveira16} (albeit, without the logical ingredient), when $H$ can be defined as the union of $k$ digraphs belonging to some allowed collection~${\cal A}$ of digraphs, the natural choice that we consider is $\kappa(H)=k$. Note that in~\cite{Oliveira16}, the collection ${\cal A}$ consists of all directed paths. 

Our main positive result is an \XP algorithm when ${\cal A}$ contains directed paths and stars (of arbitrary size). More precisely, we provide an \XP algorithm to solve {\sc Subdigraph Isomorphism} parameterized by $w$ (the directed treewidth of the input graph) and the number of directed paths and oriented stars whose union yields the desired subdigraph $H$ (cf. \autoref{thm:our-XP-algo} for the precise statement).

On the negative side, we provide a number of hardness results showing that, as far as $H$ deviates slightly from being definable by the union of few paths or stars, then {\sc Subdigraph Isomorphism} is \NP-complete even restricted to digraphs of bounded directed treewidth (at most two, and zero in almost all cases); see \autoref{theorem:hardness results} for the precise statement and \autoref{fig:hardness-digraphs-examples} for an illustration of some of the graphs $H$ for which we prove \NP-completeness. Note that the digraphs depicted from \autoref{fig:hardness-digraphs-examples}\textbf{(c)} to \autoref{fig:hardness-digraphs-examples}\textbf{(f)} are ``close'' to being definable by the union of few paths or stars, in the sense that they can be defined by the union of few digraphs that are ``close'' to directed paths or stars, namely an antidirected path in \autoref{fig:hardness-digraphs-examples}\textbf{(c)}, a digraph obtained from a (large) star by attaching small stars to each leaf in \autoref{fig:hardness-digraphs-examples}\textbf{(d)} (the hardness result also holds if one removes the root $s$), small subdivisions of exactly two stars in \autoref{fig:hardness-digraphs-examples}\textbf{(e)}, or a ``caterpillar-like'' digraph in \autoref{fig:hardness-digraphs-examples}\textbf{(f)}. The cases depicted in \autoref{fig:hardness-digraphs-examples}\textbf{(a)} and \autoref{fig:hardness-digraphs-examples}\textbf{(b)} will be discussed in the next section.

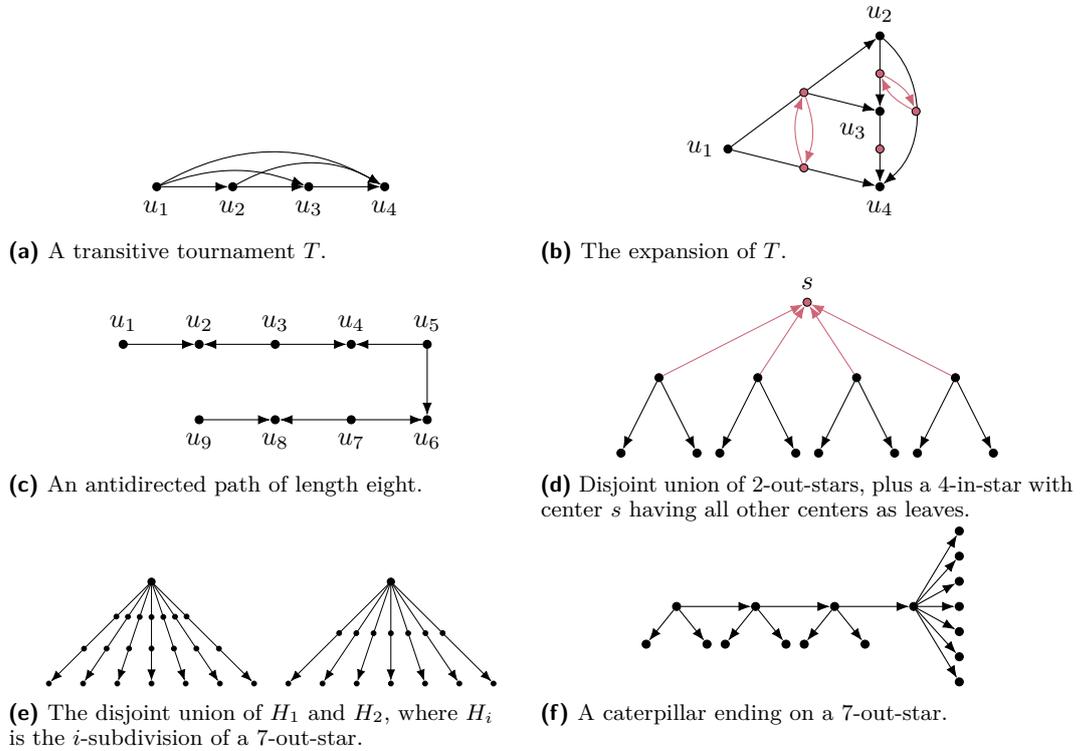
\begin{figure}[h]
\centering
\begin{subfigure}[t]{.5\textwidth}
\centering
\scalebox{1}{\begin{tikzpicture}
\foreach \i in {1,...,4}
	\node[blackvertex, scale=.6, label=-90:$u_\i$] (u\i) at (\i, 0) {};
\draw[arrow] (u1) to (u2);
\draw[arrow] (u2) to (u3);
\draw[arrow] (u3) to (u4);
\draw[arrow] (u1) to [bend left = 20] (u3);
\draw[arrow] (u1) to [bend left = 30] (u4);
\draw[arrow] (u2) to [bend left = 30] (u4);
%\node[rectangle, draw, fit=(u1)(u4), label=-90:(a), xscale=1.2, yscale=3] {};
\end{tikzpicture}}%
\caption{A transitive tournament $T$.}
\label{subfigure:transitive-tournament}
\end{subfigure}%\hspace{1cm}
\begin{subfigure}[t]{.5\textwidth}
\centering
\scalebox{1}{\begin{tikzpicture}
\node[blackvertex, scale=.6, label=180:$u_1$] (u1) at (0,0) {};
\node[blackvertex, scale=.6, label=90:$u_2$] (u2) at (2,1.5) {};
\node[blackvertex, scale=.6, label=210:$u_3$] (u3) at ($(u2) + (0,-1)$) {};
\node[blackvertex, scale=.6, label=-90:$u_4$] (u4) at ($(u3) + (0,-1)$) {};

\draw[arrow] (u1) to node [redvertex, scale=.6, midway] (in1) {} (u2);
\draw[arrow] (in1) to (u3);
\draw[arrow] (u1) to node [redvertex, scale=.6, midway] (in2) {} (u4);

\draw[arrow] (u2) to node [redvertex, scale=.6, midway] (in3) {} (u3);
\draw[arrow] (u2) to [bend left = 50] node (in4) [redvertex, scale=.6, midway] (in4) {} (u4);
\draw[arrow] (u3) to node [redvertex, scale=.6, midway] {} (u4);
\draw[arrow, goodred] (in1) to [bend left = 20] (in2);
\draw[arrow, goodred] (in2) to [bend left = 20] (in1);

\draw[arrow, goodred] (in3) to [bend left = 20] (in4);
\draw[arrow, goodred] (in4) to [bend left = 20] (in3);
\end{tikzpicture}}%
\caption{The expansion of $T$.}
\label{subfigure:expansion-of-transitive-tournament}
\end{subfigure}%

\begin{subfigure}[t]{.5\textwidth}
\centering
\scalebox{1}{\begin{tikzpicture}
\foreach \i in {1,...,5}
	\node[blackvertex, scale=.6, label=90:$u_\i$] (u\i) at (\i, 0) {};

\node[blackvertex, scale=.6, label=-90:$u_6$] (u6) at ($(u5) + (0,-1)$) {};
\node[blackvertex, scale=.6, label=-90:$u_7$] (u7) at ($(u4) + (0,-1)$) {};
\node[blackvertex, scale=.6, label=-90:$u_8$] (u8) at ($(u3) + (0,-1)$) {};
\node[blackvertex, scale=.6, label=-90:$u_9$] (u9) at ($(u2) + (0,-1)$) {};

\draw[arrow] (u1) to (u2);
\draw[arrow] (u3) to (u2);
\draw[arrow] (u3) to (u4);
\draw[arrow] (u5) to (u4);
\draw[arrow] (u5) to (u6);
\draw[arrow] (u7) to (u6);
\draw[arrow] (u7) to (u8);
\draw[arrow] (u9) to (u8);
%\node[rectangle, draw, fit=(u1)(u6), label=-90:(c), xscale=1.2, yscale=1.4] {};
\end{tikzpicture}%
}%
\caption{An antidirected path of length eight.}
\label{subfigure:antidirected-path}
\end{subfigure}%
\begin{subfigure}[t]{.5\textwidth}
\centering
\scalebox{1}{\begin{tikzpicture}
\node[redvertex, scale=.6, label=90:$s$] (s) at (3.25, 1) {};
\foreach \i in {1,...,4}{
	\node[blackvertex, scale=.6] (c\i) at (1.3*\i, 0) {};
	\node[blackvertex, scale=.6] (s1\i) at ($(c\i) + (-.5, -1)$) {};
	\node[blackvertex, scale=.6] (s2\i) at ($(c\i) + (.5, -1)$) {};
	\draw[arrow] (c\i) to (s1\i);
	\draw[arrow] (c\i) to (s2\i);
	\draw[arrow, goodred] (c\i) to (s);
}
%\node[rectangle, draw, fit=(s)(s11)(s24), label=-90:(d), xscale=1, yscale=1.2] {};
\end{tikzpicture}%
}%
\caption{Disjoint union of $2$-out-stars, plus a $4$-in-star with center $s$ having all other centers as leaves.}
\label{subfigure:many-stars-plus-big-star}
\end{subfigure}%

\begin{subfigure}[t]{.5\textwidth}
\centering
\scalebox{.9}{\begin{tikzpicture}
\node[blackvertex, scale=.6] (s1) at (0,0) {};
\node[blackvertex, scale=.6] (s2) at (3.5,0) {};
\foreach \i in {-3,...,3} {
	\node[blackvertex, scale=.4] (u\i) at ($(s1) + (\i/2, -1.5)$) {};
	\draw[arrow] (s1) to node[blackvertex, pos=.33, scale=.4] {} node[blackvertex, pos=.66, scale=.4] {} (u\i);
	\node[blackvertex, scale=.4] (v\i) at ($(s2) + (\i/2, -1.5)$) {};
	\draw[arrow] (s2) to node[blackvertex, pos=.5, scale=.4] {} (v\i);
}
%\node[rectangle, draw, fit=(s1)(s2)(u-3)(v3), label=-90:(e), xscale=1, yscale=1.35] {};
\end{tikzpicture}}%
\caption{The disjoint union of $H_1$ and $H_2$, where $H_i$\\is the $i$-subdivision of a $7$-out-star.}
\label{subfigure:stars-subdivisions}
\end{subfigure}%
\begin{subfigure}[t]{.5\textwidth}
\centering
\scalebox{1}{\begin{tikzpicture}[xscale=.8]
\foreach \i in {1,...,3} {
	\node[blackvertex, scale=.6] (u\i) at (1.3*\i,0) {};
	\node[blackvertex, scale=.6] (s1) at ($(u\i) + (-.5, -.5)$) {};
	\node[blackvertex, scale=.6] (s2) at ($(u\i) + (.5, -.5)$) {};
	\draw[arrow] (u\i) to (s1);
	\draw[arrow] (u\i) to (s2);
}
\node[blackvertex, scale=.6] (u4) at (1.3*4,0) {};
\foreach \i in {-3, ..., 3} {
	\node[blackvertex, scale=.6] (v\i) at ($(u4) + (.75, \i/3)$) {};
	\draw[arrow] (u4) to (v\i);
}
\draw[arrow] (u1) to (u2);
\draw[arrow] (u2) to (u3);
\draw[arrow] (u3) to (u4);
\end{tikzpicture}}%
\caption{A caterpillar ending on a $7$-out-star.}
\label{subfigure:caterpillar}
\end{subfigure}%
\caption{Examples of digraphs from the classes mentioned in \autoref{theorem:hardness results}.
Items \lipItem{3.} and \lipItem{4.} of \autoref{theorem:hardness results} are represented by \textbf{(d)}: the former when $s$ is deleted from the digraph, and the latter with the digraph as it is.}
\label{fig:hardness-digraphs-examples}
\end{figure}

\medskip
\noindent\textbf{Organization}.
In \autoref{sec:overview} we provide an overview of the techniques that we use to prove our positive and negative results. 
In \autoref{sec:prelim}, we provide the preliminaries about digraphs, directed treewidth, and formally define the classes of digraphs relevant to this paper. In \autoref{sec:algorithm} (resp. \autoref{sec:hardness_results}) we state and prove our algorithmic (resp. hardness)  results.
We conclude the article in \autoref{sec:conclusions} with some directions for further research.

\section{Overview of our techniques} 
\label{sec:overview}

In this section we describe in more detail our results and present the main ideas involved in the proofs. Full details can be found in the next sections. 

\medskip
\noindent\textbf{The \XP algorithm}. On the positive side, as mentioned in~\autoref{sec:intro}, we give an \XP algorithm for the case when $H$ is formed by the disjoint union of stars and internally disjoint paths between the centers of those stars (note that unions of paths are also captured in this way, by considering stars with just one vertex as the endpoints), which we name {\sc \pname}.
Our \XP algorithm is parameterized by the number of stars and the directed treewidth of the given digraph.
The algorithm consists of two parts.
First, we notice that if the goal is to find only the stars, one can reduce the problem to an integer system of inequations, where the inequations are given by the size of the neighborhoods of the vertices candidates for the centers of the stars, and the size of the stars that we want to build.
This part runs in \FPT time parameterized by the number of stars after the candidates for the centers are given, using the fact that an integer system of linear inequations can be solved in \FPT time parameterized by the number of variables~\cite{Lenstra83}.
Additionally, we can add the following query to the algorithm: if, for every candidate vertex $v \in D$, we are given a number $\slack(v)$ representing how many vertices in the usable neighborhood of $v$ we want to \emph{avoid} using for the stars, can we find the stars while respecting the slacks?
This is needed for the next part of the algorithm, where the goal is to route the paths while avoiding sufficiently many vertices of the usable neighborhoods such that the stars can be built.

More precisely, in the second part, we introduce a new version of the \textsc{Directed Disjoint Paths} problem where, in addition to the terminal set $\{s_1, t_1, \ldots, s_r, t_r\}$ that we want to connect with pairwise vertex-disjoint paths, we are given subsets $X_1, \ldots, X_k \subseteq V(D)$ and integers $x_1, \ldots, x_k$, and the goal is to route the paths while using at most $x_i$ vertices of $X_i$.
We call this the \textsc{\pddp} (\textsc{\pddpshort}) problem.
Applying the framework of Johnson et al.~\cite{Johnson.Robertson.Seymour.Thomas.01} (see also Lopes and Sau~\cite{LopesS22} for another application of this framework), we show that \textsc{\pddpshort} is \XP with parameters $r$ and $k$.
As in the case of the \XP algorithm for \textsc{DDP} in~\cite{Johnson.Robertson.Seymour.Thomas.01}, we exploit the bounded breakability of paths to obtain a dynamic programming algorithm for our problem.
We remark that directed treewidth appears only in this part of the algorithm.

Given a set of candidates for the centers of the stars, for each choice of the slacks for those vertices, we build an instance of \textsc{\pddpshort} where the sets $X_i$ are built from the neighborhoods of the candidates, and solve it in \XP time with parameters $r$, $k$, and the directed treewidth of the given digraph.
Since the number of possible choices of for the slacks is $\Ocal(n^k)$, this strategy yields an \XP algorithm.
Since the end goal is an \XP algorithm, we can pay the cost of guessing where to embed the centers of the stars in the digraph and thus use the rooted version to solve the unrooted one. Note that, with this approach, an \XP algorithm is the best we can hope for, since {\sc Directed Disjoint Paths}, which we solve as a particular case of our algorithm, is known to be ${\sf W}[1]$-hard (hence, unlikely to be \FPT) even restricted to DAGs~\cite{Slivkins.03}.

\medskip
\noindent\textbf{Hardness results}. Our hardness results are itemized in \autoref{theorem:hardness results} and follow from several distinct reductions, which we proceed to explain  here. Namely, we consider the following cases for the choice of the target digraph $H$ in {\sc Subdigraph Isomorphism.}
We say that $H$ is an \emph{antidirected path} if $H$ is an orientation of an undirected path where the arcs alternate between forward and backwards arc along the path.
In this case, we improve on a result from Bang-Jensen et al.~\cite[Theorem 2.2]{BANGJENSEN201768} which proves that the rooted version of the problem is \NP-complete.
Their reduction does not generate a DAG, however, and we show how to modify it  in order to do so.
In short, by increasing the size of the antidirected paths added to the constructed digraph $D$ in the reduction, we ensure that no two vertices which are not sinks nor sources in $D$ are within a directed path of $D$.

An \emph{out-star} (\emph{in-star}) is an orientation of an undirected star where all arcs are oriented away from (towards) the center of the star. 
If $H$ is formed by the disjoint union of out-stars (in-stars) with two leaves, we provide a novel reduction from a matching problem in bipartite graphs, introduced and shown to be \NP-complete by Plaisted and Zaks~\cite{PLAISTED198065}.
In this matching problem, we are given a bipartite graph $G$ with parts $V_1$ and $V_2$ together with partitions $\mathcal{P}_1$ and $\mathcal{P}_2$ of $V_1$ and $V_2$, respectively.
The goal is to decide if $G$ has a perfect matching $M$ such that no two distinct edges $\{a_1, b_1\}, \{a_2, b_2\}\in M$ have the property that $a_1, a_2$ are in the same part of $\mathcal{P}_1$ or $b_1, b_2$ are in the same part of $\mathcal{P}_2$.
A perfect matching with such properties is said to be \emph{consistent}.
If the goal is to prove the result for out-stars, we can begin as follows. 
For each $e \in E(G)$ with endpoints $u,v$, one can add the arc $\{u,v\}$ to the digraph $D$ plus a new vertex $v_e$ with out-neighbors $u$ and $v$, then search for $|V_1| = |V_2|$ out-stars in $D$.
This does not guarantee that the leaves of the out-stars are selecting edges of $G$ forming a consistent matching.
Thus, we exploit the extra properties of the results of~\cite{PLAISTED198065}.

The crucial property is that in~\cite{PLAISTED198065} the authors show that the problem remains hard even if every part of $\mathcal{P}_1$ and every part of $\mathcal{P}_2$ has size two.
This allows our reduction to distinguish four possible configurations for each pair $A, B$ with $A \in \mathcal{P}_1$ and $B \in \mathcal{P}_2$, depending on how many edges there are between $A$ and $B$.
For each configuration with two or more edges between the sets, we add some new stars to the digraph and increase the number of stars we are going to search for in the instance of {\sc Subdigraph Isomorphism}.
Then we show that the stars that we added in the configurations are sufficient to select a consistent matching in $G$ if the appropriate number of stars is found.

For $i \in [2]$, let $S_i$ be the digraph formed by subdividing $i$ times every arc of an out-star.
For the case when $H$ is formed by the disjoint union of a copy of $S_1$ and a copy of $S_2$, we provide a reduction from a particular case of $3$-\textsc{SAT} in which every literal appears {\sl exactly} twice in the clauses.
This version of $3$-\textsc{SAT} was shown to be \NP-complete by Berman et al.~\cite{Berman2003}.
The specification on the number of times each literal can appear is key for the reduction to work.
We start with a vertex $s$ and add to the digraph $2n$ (not necessarily disjoint) paths of length three leaving $s$, where $n$ is the number of variables.
Each pair of paths represents a choice of a truth value to a variable.
The two last vertices of each path are associated with the two positive and two negated occurrences of a variable in the clauses, and each path leaving $s$ uses the two vertices associated with the two positive or the two negative occurrences of a variable in the clauses.
Thus the second clause, built from a structure associated with the clauses, can only have leaves in the two vertices \textsl{not} selected by the path leaving $s$, for each variable.
The star with center $s$ acts as a \emph{selector} for the variables, and a star with a center $c$ associated with the clauses acts as a \emph{validator} for the clauses.

\medskip
\noindent\textbf{The role of breakability}. In all the aforementioned cases in our hardness reductions, the constructed digraphs are DAGs.
The key ``separation'' notion of directed treewidth is that of \emph{$w$-guarded} sets (see \autoref{sec:prelim} for the precise definition).
In short, given an integer $w$, a set of vertices $X$ of a digraph $D$ is $w$-guarded if there is a set $Z \subseteq V(D) \setminus S$ of size at most $w$ such that no path of $D$ starts in $X$, intersects $V(D) \setminus (X \cup Z)$, and returns to $X$ without using a vertex of $Z$.

In many cases, parameterized algorithms exploiting the structure of digraphs with bounded directed treewidth begin by showing that, when the goal is to find some subdigraph $H \subseteq D$, that $H$ cannot be split into too many pieces by a $w$-guarded set.
In other words, if $X$ is $w$-guarded, the maximum number of weak components in the subdigraph induced by $V(H) \cap X$ is bounded by some function of $w$.
This approach is used, for example, in~\cite{Johnson.Robertson.Seymour.Thomas.01,LopesS22}.
We formalize this notion by saying that $H$ has \emph{bounded breakability} (see \autoref{def:breakability} for the precise definition).
A priori, one might expect that bounded breakability is a strong enough condition to ensure the existence of \XP algorithms for \NP-hard problems restricted to digraphs of bounded directed treewidth.
Somehow surprisingly, we show that this is not the case, even if the target digraph also has bounded maximum degree together with bounded breakability.

Defining a class $\mathcal{H}$ of digraphs with bounded breakability, bounded maximum degree, and such that {\sc Subdigraph Isomorphism} is \NP-complete with respect to target digraphs in $\mathcal{H}$ turned out to be no easy task.
We do so by augmenting a classical construction that is used to bound the maximum degree of digraphs.
Given a digraph $D$, we define the \emph{expansion} of $D$ as the digraph $D'$ built following a series of steps  (cf. \autoref{fig:hardness-digraphs-examples}\textbf{(a)} and \autoref{fig:hardness-digraphs-examples}\textbf{(b)}).
Namely, and informally, we apply the classical procedure of replacing all the arcs leaving a vertex $u \in V(D)$ by an out-arborescence with root $u$ (that is, the orientation of a tree with root $r$ where all arcs are oriented away from $u$).
Then, we do the same for all vertices of $V(D)$ which have in-degree greater than two, but now for incoming arcs.
This bounds the maximum degree of $D'$, but is not enough to bound its breakability.

For the breakability, we add arcs to $D'$ in order to ensure that the local out-  and in-arborescences added to $D'$ when a vertex $u$ is processed are augmented to ``almost strong'' digraphs.
We need to be careful here: although an increase in the connectivity is needed to ensure that for every $u,v \in V(D')$ there is a path from $u$ to $v$ or from $v$ to $u$ (or both) when $D'$ is the expansion of a transitive tournament, we need to avoid creating paths between the ``original'' vertices $x,y$ of $D'$ (that is, $x,y \in V(D)$) when there are no paths paths between $x$ and $y$ in $D$.
Informally, the connectivity needs to increase, but not by too much that we lose control on the directed treewidth of expansions of DAGs.
This is important to maintain since the goal is to prove hardness for expansions of tournaments when the base digraph has directed treewidth at most three.

Let $T'$ be the expansion of a transitive tournament $T$.
Proving that {\sc Subdigraph Isomorphism} is hard when $H$ is the expansion of a transitive tournament is easy. We simply follow the same idea as in the proof of hardness for transitive tournaments.
Proving that the breakability of $T'$ is bounded, however, is much harder.
We first show that, given a vertex $v \in V(T')$, one can define a partition $V(T') \setminus \{v\}$ into four sets $\{X^v_i \mid i \in [4]\}$ with the following property: for any $x,y \in X^v_i$, there is a path from $x$ to $y$ or from $y$ to $x$ {\sl avoiding}  $v$.
From here, the proof follows by induction.
In short, given a vertex $v$ in a guard $Z$ of size $w$, we split $V(T')$ as mentioned into sets $\{X^v_i \mid i \in [4]\}$.
Since $v$ is part of the guard, each $V(T') \cap X^v_i$ is $w-1$-guarded.
By induction, $V(T') \cap X^v_i \cap X$ has at most $4^{w-1}$ weak components when $X$ is a $w$-guarded set, which implies that the breakability of $T'$ is at most~$4^w$.

\section{Definitions and preliminaries}
\label{sec:prelim}
In this section we provide basic preliminaries about digraphs, parameterized complexity, and directed treewidth.

\subsection{Digraphs, directed treewidth, and parametrized complexity}
For basic background on graph theory we refer the reader to~\cite{Graph.Theory}.
Since in this article we mainly work with digraphs, we focus on basic definitions of
digraphs, often skipping their undirected counterparts.
Given a digraph $D$ we denote by $V(D)$ and $E(D)$ the sets of vertices and arcs of $G$, respectively.
Unless stated otherwise, all involved digraphs are simple, i.e., they have neither loops nor parallel arcs with the same orientation.
Given $X \subseteq V(D)$, we denote by $G \setminus X$ the digraph resulting from
removing every vertex of $X$ from $D$.
We denote by $D[X]$ the subdigraph of $D$ \emph{induced} by $X$.

If $e$ is an arc of a digraph from a vertex $u$ to a vertex $v$, we say that $e$ has
\emph{endpoints} $u$ and $v$, that $e$ is \emph{incident} to $u$ and $v$, and that $e$ is
\emph{oriented} from $u$ to $v$. We may refer to $e$ as the ordered pair $\{u,v\}$.
In this case, $u$ is the \emph{tail} of $e$ and $v$ is the \emph{head} of $e$.
We also say that $e$ is \emph{leaving} $u$ and \emph{reaching} $v$, and that $u$ and $v$
are \emph{adjacent}.
A \emph{clique} in a (di)graph $D$ is a set of pairwise adjacent vertices of $D$, and an
\emph{independent set} is a set of pairwise non-adjacent vertices of $D$.
A pair of arcs is \emph{adjacent} if they share an endpoint.
A \emph{matching} is a set of pairwise non-adjacent edges (arcs) of a (di)graph.
For $A, B \subseteq V(D)$, we say that $M$ is a \emph{matching from $A$ to $B$} if every arc of $M$ has tail in $A$ and head in $B$.
The \emph{in-degree} (resp. \emph{out-degree}) of a vertex $v$ in a digraph $D$ is the
number of arcs with head (resp. tail) $v$.
We denote the in-degree and out-degree of $v$ by $\deg^-_D(v)$ and $\deg^+_D(v)$, respectively.

The \emph{in-neighborhood} $N^-_D(v)$ of $v$ is the set $\{u \in V(D) \mid (u,v) \in
E(D)\}$, and the \emph{out-neighborhood} $N^+_D(v)$ is the set $\{u \in V(D) \mid \{v,u\}
\in E(D)\}$.
We say that $u$ is an \emph{in-neighbor} of $v$ if $u \in N^-_D(v)$ and that $u$ is an
\emph{out-neighbor} of $v$ if $u \in N^+_D(v)$.
We extend these notations to sets of vertices: given $X \subseteq V(D)$, we define
$N^-_D(X) = (\bigcup_{v \in X}N^-_D(v)) \setminus X$ and $N^+_D(X) = (\bigcup_{v \in
X}N^+_D(v)) \setminus X$.
A vertex $v$ is a \emph{source} (\emph{sink}) in $D$ if $\deg^-_D(v) = 0$ ($\deg^+_D(v) = 0$).

A \emph{walk} in a digraph $D$ is an alternating sequence $W$ of vertices and arcs that
starts and ends with a vertex, and such that for every arc $(u,v)$ in the walk, vertex
$u$ (resp. vertex $v$) is the element right before (resp. right after) arc $(u,v)$ in $W$.
If the first vertex in a walk is $u$ and the last one is $v$, then we say this is a
\emph{walk from $u$ to $v$}.
A \emph{path} is a walk without repetition of vertices nor arcs.
The \emph{length} of a path is the number of arcs in the path.

An \emph{orientation} of an undirected graph $G$ is a digraph $D$ obtained from $G$ by choosing an orientation for each edge $e \in E(G)$.
The undirected graph $G$ formed by ignoring the orientation of the arcs of a digraph $D$ is the \emph{underlying graph} of $D$.

A digraph $D$ is \emph{strongly connected} if, for every pair of vertices $u,v \in V(D)$, there is a walk from $u$ to $v$ and a walk from $v$ to $u$ in $D$.
We say that $D$ is \emph{weakly connected} if the underlying graph of $D$ is connected.
A \emph{separator} of $D$ is a set $S \subsetneq V(D)$ such that $D \setminus S$ is not strongly connected.
If $|V(D)| \geq k+1$ and $k$ is the minimum size of a separator of $D$, we say that $D$ is \emph{$k$-strongly connected}.
A \emph{strong component} of $D$ is a maximal induced subdigraph of $D$ that is strongly connected, and a \emph{weak component} of $D$ is a maximal induced subdigraph of $D$ that is weakly connected.

An \emph{isomorphism} between digraphs $D$ and $D'$ is a bijective mapping $f: V(D) \to V(D')$ such that $\{f(u),f(v)\} \in E(D')$ if and only if $\{u,v\}\in E(D)$.
We write $D \cong D'$ to say that digraphs $D$ and $D'$ are \emph{isomorphic}, i.e., that there is an isomorphism between $D$ and $D'$.

A \emph{bipartite} (di)graph $G$ is a (di)graph with vertex partition $\{V_1, V_2\}$ such that every edge (arc) of $G$ has one endpoint $V_1$ and one endpoint in $V_2$.
A \emph{perfect matching} in a (di)graph is a matching such that every vertex appears in exactly on edge (arc) of the matching.

For an integer $\ell \geq 1$, by doing an \emph{$\ell$-subdivision} of an arc $\{u,v\} \in E(D)$ we delete $(u,v)$ from $D$ and add to the resulting digraph a path with $2+\ell$ vertices from $u$ to $v$.

For a positive integer $k$, we denote by $[k]$ the set of integers $\{1, \ldots, k\}$ and by $2^{[k]}$ the collection of all subsets of $[k]$.

Unless stated otherwise, $D$ will always stand for a digraph, $G$ for an undirected graph, and $n = |V(D)|$ and $m = |E(D)|$ when $D$ is the input digraph of some algorithm.

\subsection{Parameterized complexity}

We refer the reader to~\cite{DF13,CyganFKLMPPS15} for basic background on parameterized complexity, and we recall here only the definitions used in this article. A \emph{parameterized problem} is a language $L \subseteq \Sigma^* \times \mathbb{N}$.  For an instance $I=(x,k) \in \Sigma^* \times \mathbb{N}$, $k$ is called the \emph{parameter}.

A parameterized problem $L$ is \emph{fixed-parameter tractable} ({\sf FPT}) if there exists an algorithm $\mathcal{A}$, a computable function $f$, and a constant $c$ such that given an instance $I=(x,k)$, $\mathcal{A}$   (called an {\sf FPT} \emph{algorithm}) correctly decides whether $I \in L$ in time bounded by $f(k) \cdot |I|^c$. For instance, the \textsc{Vertex Cover} problem parameterized by the size of the solution is {\sf FPT}.
  
A parameterized problem $L$ is in {\sf XP} if there exists an algorithm $\mathcal{A}$ and two computable functions $f$ and $g$ such that given an instance $I=(x,k)$, $\mathcal{A}$  (called an {\sf XP} \emph{algorithm}) correctly decides whether $I \in L$ in time bounded by $f(k) \cdot |I|^{g(k)}$. For instance,  the \textsc{Clique} problem parameterized by the size of the solution is in  {\sf XP}.

Within parameterized problems, the class {\sf W}[1] may be seen as the parameterized equivalent to the class {\sf NP} of classical decision problems. Without entering into details (see~\cite{DF13,CyganFKLMPPS15} for the formal definitions), a parameterized problem being {\sf W}[1]-\emph{hard} can be seen as a strong evidence that this problem is {\sl not} {\sf FPT}.
The canonical example of {\sf W}[1]-hard problem is \textsc{Clique}  parameterized by the size of the solution.

\subsection{Arboreal decompositions, arborescences, and directed treewidth}\label{section:arboreal_decomp_obstructions}

By an \emph{out-arborescence} $R$ with \emph{root} $r_0$, we mean an orientation of a tree such that $R$ contains a path from $r_0$ to every other vertex of the tree.
If a vertex $v$ of $R$ has out-degree zero, we say that $v$ is a \emph{leaf} of $R$. We now define guarded sets and arboreal decompositions of digraphs. 
An \emph{ir-arborescence} with root $r_0$ is an out-arborescence with root $r_0$ where the directions of all arcs are reversed.
A leaf of an in-arborescence is a vertex of in-degree zero.

\begin{definition}[$Z$-guarded and $w$-guarded sets]
\label{def:guarded-sets}
Let $D$ be a digraph, $Z \subseteq V(D)$, and $S \subseteq V(D) \setminus Z$.
We say that $S$ is \emph{$Z$-guarded} if there is no directed walk in $D \setminus Z$ with first and last vertices in $S$ that uses a vertex of $V(D) \setminus (Z \cup S)$.
For an integer $w \geq 0$, we say that $S$ is \emph{$w$-guarded} if $S$ is $X$-guarded by some $X \subseteq V(D)$ with $|X| \leq w$.
\end{definition}
\noindent That is, informally speaking, a set $S$ is $Z$-guarded if whenever a walk starting in $S$ leaves $S$, it is impossible to come back to $S$ without visiting a vertex in $Z$. See~\autoref{fig:z_guarded_set} for an illustration of a $Z$-guarded set.
\begin{figure}[h!]
\centering
\scalebox{.75}{\begin{tikzpicture}[scale=1]
\draw[rounded corners] (0,4) rectangle  (5,5) node [above,xshift=-2.5cm] {$V(D)\setminus (Z \cup S)$} ;
\draw[rounded corners] (0,0) rectangle (5,1) node [below,yshift=-1cm,xshift=-2.5cm] {$S$}; 
\draw[rounded corners] (1,2) rectangle (4,3) node [above,xshift=.5cm,yshift=-.75cm] {$Z$};
%\node (P-s) at (1.5, 4.2) {$S$};
%\node (P-v) at (1.5, -0.2) {$V(D) - (Z \cup S)$};
%\node (P-z) at (5.7, 2) {$Z$};
\node[blackvertex,scale=.5] (P-a) at (0.5,0.5) {};
\node (P-b) at (-0.5,2.5) {};
\node (P-c) at (0.25,4.5) {};
\node (P-d) at (0.85,4.75) {};
\node[blackvertex,label=$u$,scale=.5] (P-e) at (1.25,4.5) {};
\node (P-f) at (1.5,2.5) {};
\node (P-g) at (0.75,1.5) {};
\node[blackvertex,scale=.5] (P-h) at (1.5,0.5) {};

\draw[thick, -{Latex[length=3mm, width=2mm]}] plot[smooth] coordinates {(P-a) (P-b) (P-c) (P-d) (P-e) (P-f) (P-g) (P-h) };

\node[blackvertex,scale=.5] (P-a1) at (4,0.5) {};
\node (P-b1) at (5,2.5) {};
\node (P-c1) at (4.25,4.5) {};
\node (P-d1) at (3.75,4.75) {};
\node[blackvertex,label=$v$,scale=.5] (P-e1) at (2.5,4.5) {};
\node (P-f1) at (3,2.5) {};
\node (P-g1) at (2.25,1.5) {};
\node[blackvertex,scale=.5] (P-h1) at (3,0.5) {};

\draw[thick, -{Latex[length=3mm, width=2mm]}] plot[smooth] coordinates {(P-h1) (P-g1) (P-f1) (P-e1) (P-d1) (P-c1) (P-b1) (P-a1)};

\draw[dashed, -{Latex[length=2mm, width=2mm]}] (P-e) to [bend right =30] (P-e1);

%\draw[edge,line width=1.2] (P-us) to [bend right = 35] (P-ur);
%\draw[edge,line width=1.2] (P-vr) to [bend right = 35] (P-uz);
%\draw[edge,line width=1.2] (P-vz) to [bend right = 35] (P-vs);
%\draw[->] plot[smooth]

\end{tikzpicture}%
}%
\caption{A $Z$-guarded set $S$. The dashed line indicates that there is no path from $u$ to $v$ in $V(D) \setminus (Z \cup S)$.}
\label{fig:z_guarded_set}
\end{figure}
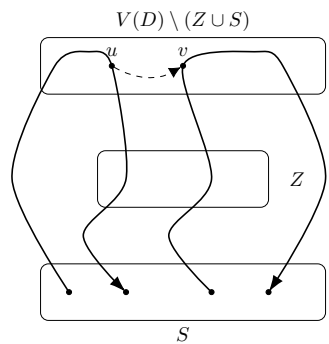
If a set $S$ is $Z$-guarded, we may also say that $Z$ is a \emph{guard} for $S$. We remark that in~\cite{Johnson.Robertson.Seymour.Thomas.01}, the authors use the terminology of $Z$-\emph{normal} sets instead of $Z$-guarded sets.

Let $R$ be an out-arborescence, $r \in V(R)$, $e \in E(R)$, and $r'$ be the head of $e$.
We say that $r > e$ if there is a path from $r'$ to $r$ in $R$. 
We also say that $e \sim r$ if $r$ is the head or the tail of $e$.
To define the treewidth of directed graphs, we first need to introduce arboreal decompositions.
\begin{definition}[Arboreal decomposition]
An \emph{arboreal decomposition} $\beta$ of a digraph $D$ is a triple $(R,\mathcal{X},\mathcal{W})$ where $R$ is an out-arborescence, $\mathcal{X} = \{X_e : e \in E(R)\}$, $\mathcal{W} = \{W_r : r \in V(R)\}$, and $\mathcal{X},\mathcal{W}$ are collections of sets of vertices of $D$ (called \emph{bags}) such that
\begin{enumerate}
\item[\textbf{(i)}] $\mathcal{W}$ is a partition of $V(D)$ into non-empty sets, and
\item[\textbf{(ii)}] if $e \in E(R)$, then $\bigcup\{W_r : r \in V(R) \text{ and } r > e\}$ is $X_e$-guarded.
\end{enumerate}
We also say that $r$ is a \emph{leaf} of $(R,\mathcal{X,W})$ if $r$ has out-degree zero in $R$.
\end{definition}
The left hand side of~\autoref{fig:arboreal_decomposition} contains an example of a digraph $D$, while the right hand side shows an arboreal decomposition for it.
In the illustration of the arboreal decomposition, squares are guards $X_e$ and circles are bags of vertices $W_r$.
For example, consider the edge $e \in E(R)$ with $X_e = \{b,c\}$ from the bag $W_1$ to the bag $W_2$.
Then $\bigcup\{W_r : r \in V(R) \text{ and } r > e\} = V(D) \setminus \{a\}$ and, by item \textbf{(ii)} described above, this set must be $\{b,c\}$-guarded since $X_e = \{b,c\}$.
In other words, there cannot be a walk in $D \setminus \{b,c\}$ starting and ending in $V(D) \setminus \{a\}$ using a vertex of $\{a\}$.
This is true in $D$ since every path reaching $\{a\}$ from the remaining of the graph must do so through vertices $b$ or $c$.
The reader is encouraged to verify the same properties for the other guards in the decomposition.
\begin{figure}[ht]
\centering
\begin{subfigure}{.35\textwidth}
\scalebox{.85}{\begin{tikzpicture}

\foreach \x/\y/\name/\lpos in {
	0/5/a/90,
	-1.5/4/b/90, 1.5/4/c/90}
	\node[blackvertex, label=\lpos:{$\name$},scale=.5] (P-\name) at (\x,1.2*\y) {$\name$};

\foreach \x/\y/\name/\lpos in {
	-2.25/3/d/120, -0.75/3/e/60, 0.75/3/f/120, 2.25/3/g/60}
	\node[blackvertex, label={[label distance= -.1cm]\lpos:{$\name$}},scale=.5] (P-\name) at (\x,1.2*\y) {$\name$};

\foreach \x/\y in {
a/b, a/c, b/c, b/d, b/e, c/b, c/f, c/g, d/b, e/b, f/c, g/c, d/e, f/g}
	\draw[edge, {Latex[length=2mm, width=2mm]}-{Latex[length=2mm, width=2mm]}, line width = 1, shorten >= .1cm, shorten <= .1cm] (P-\x) to (P-\y);

\end{tikzpicture}%
}
\end{subfigure}\hspace{1cm}
\begin{subfigure}{.35\textwidth}
\scalebox{.75}{\begin{tikzpicture}%[rectangle/.style={regular polygon,regular polygon sides=4}]
	\foreach \x/\y/\name/\idn/\lbl/\lblpos in {
	0/5/w1/a/{W_1}/180,
	0/2.5/w2/{b,c}/{W_2}/270,
	-2.5/2/w3/{d,e}/{W_3}/90, 2.5/2/w4/{f,g}/{W_4}/90}
	\node[vertex, text width=.75cm, align=center, label, label=\lblpos:{\large$\lbl$}] (P-\name) at (\x,\y) {\idn};

%\foreach \x/\y/\name in {
%	0/4/g1,
%	0/4/w2,
%	-1.5/3/w3, 1.5/3/w4}
%	\node[vertex, text width=.5cm, align=center] (P-\name) at (\x,1.5*\y) {\name};
\draw[edge, line width = 1] (P-w1) to  node[midway, rectangle, fill=white,draw] {$b,c$}  (P-w2); 
\draw[edge, line width = 1] (P-w2) to  node[midway, rectangle, fill=white,draw] {$b$}  (P-w3); 
\draw[edge, line width = 1] (P-w2) to  node[midway, rectangle, fill=white,draw] {$c$}  (P-w4);

\end{tikzpicture}%
}
\end{subfigure}
\caption{A digraph $D$ and an arboreal decomposition of $D$ of width two. A bidirectional arc is used to represent a pair of arcs in both directions.}
\label{fig:arboreal_decomposition}
\end{figure}
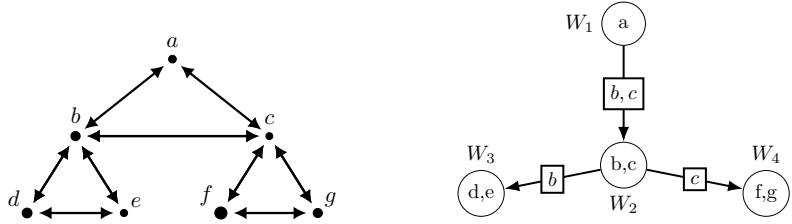

\begin{definition}[Directed treewidth]
Let $(R,\mathcal{X},\mathcal{W})$ be an arboreal decomposition of a digraph $D$.
For a vertex $r \in V(R)$, we denote by $\width(r)$ the size of the set $W_r \cup (\bigcup_{e \thicksim r}X_e)$.
The \emph{width} of $(R,\mathcal{X},\mathcal{W})$ is the least integer $k$ such that, for all $r \in V(R)$, $\width(r) \leq k+1$.
The \emph{directed treewidth}  of $D$, denoted by $\dtw(D)$, is the least integer $k$ such that $D$ has an arboreal decomposition of width $k$.
\end{definition}
\noindent We remark that DAGs have directed treewidth zero but they are \textsl{not} the only digraphs with directed treewidth zero if loops are allowed since the addition of loops to a digraph $D$ does not change the directed treewidth of $D$.

Johnson et al.~\cite{Johnson.Robertson.Seymour.Thomas.01} showed that the concept of directed treewidth is indeed a generalization of treewidth of undirected graphs.
\begin{proposition}[Johnson et al.~{\cite[2.1]{Johnson.Robertson.Seymour.Thomas.01}}]
\label{proposition:directed-tw-generalizes-undirected-tw}
Let $G$ be an undirected graph and $D$ the digraph obtained from $G$ by replacing every edge of $G$ with two arcs in opposite directions.
The the treewidth of $G$ is equal to the directed treewidth of $D$.
\end{proposition}
Thus, deciding if a digraph $D$ has directed treewidth at most $k$, for a given integer $k$, is {\NP}-complete, following the hardness of treewidth~\cite{doi:10.1137/0608024}.

Johnson et al.~\cite[3.3]{Johnson.Robertson.Seymour.Thomas.01} also gave an \XP with parameter $k$ that either produces an arboreal decomposition of width $3k+2$ of a given digraph $D$, or correctly decides that $\dtw(D) \geq k-1$, and a sketch of a proof of an \FPT algorithm with an approximation factor of $5k+10$ is given in~\cite[Theorem 9.4.4]{Classes.Directed.Graphs}.
We elect to state and use the algorithm by Campos et al.~\cite{doi:10.1137/21M1452664}, which has a better approximation ratio.
\begin{proposition}[Campos et al.~\cite{doi:10.1137/21M1452664}]
Let $D$ be a digraph and $k$ be a non-negative integer.
There is an algorithm running in time $2^{\Ocal(k \log k)}\cdot n^{\Ocal(1)}$ that either produces an arboreal decomposition of with $3k-2$ or correctly decides that $\dtw(D) \geq k-1$.
\end{proposition}

\subsection{Classes of digraphs}
In this section we define the classes of digraphs which are used in our hardness and algorithmic proofs.

\medskip
\noindent\textbf{Antidirected paths.} We say that a digraph $H$ is an \emph{antidirected path} if $H$ is an orientation of a path in which the arcs alternate between forward and backwards arcs.

\medskip
\noindent\textbf{Expansions of digraphs and tournaments.}
A \emph{tournament} is an orientation of a complete graph.
A \emph{transitive tournament} is an acyclic tournament.
A usual technique to bound the maximum degree of a digraph $D$ is to replace the arcs leaving a vertex $u$ by a vertex-minimal out-arborescence with root $u$, maximum out-degree two, and having every vertex of $N^+_D(u)$ as leaves.
For our case, this construction is not sufficient since we not only need to bound the maximum degree, but also ensure that the constructed digraphs have bounded $w$-breakability and, in general, the $w$-breakability of digraphs resulting from these constructions is not guaranteed to be bounded.
To see this, consider a source vertex $v$ of transitive tournament $T$.
Thus the out-arborescence added to the expansion of $T$ when $v$ is processed is then very large and has many vertices which are pairwise unreachable.
To find such a set one can take, for example, all vertices within the same distance of the root of the out-arborescence.

Thus, we augment this construction in order to ensure that, in the expansion of transitive tournaments, for any pair of vertices $v,w$ in the generated digraph, there is a path from $v$ to $w$ or a path from $w$ to $v$\footnote{This property resembles the defining property of \emph{unilateral digraphs}, where for each pair of vertices $v,w$ there is either a path from $v$ to $w$ or a path from $w$ to $v$.} (see the proof of \autoref{lem:breakability-of-expanded-TTk}).

Let $D$ be a digraph without loops or parallel arcs.
We define the \emph{expansion} of $D$ as the digraph $D'$ built from $D$ through the following operations.
Start with $D'= D$.
In the first part of the construction, we iteratively process every vertex $u \in V(D)$ together with its out-neighborhood in $D$.
Namely, for every $u \in V(D)$, delete every arc leaving $u$ from $D'$ and add to $D'$ a vertex-minimal out-arborescence with root $u$, maximum out-degree two and having every vertex of $N^+_D(u)$ as leaves.
Notice that $V(D) \subseteq V(D')$.
We refer to $V(D)$ as the \emph{original vertices} of $D'$ and to all other vertices of $D'$ as \emph{internal vertices}.
Now, we subdivide once every arc from $u$ to an original vertex of $D'$ if any exists.
This ensures that there is at least one internal vertex in the path from $u$ to each $v \in N^+_D(u)$ in $D'$.
Then, we replace every arc between internal vertices of $D'$ added when processing $u$ by a pair of parallel arcs in both directions.
If $u$ is an original vertex of $D'$ with $\deg^+_{D'}(u) = 2$, then we add two arcs in both directions between the two out-neighbors of $u$ in $D'$.
Finally, add a loop in $D'$ to every original vertex.
This construction implies that $D'$ has maximum out-degree three: an internal vertex $y$ can have up to four neighbors if one of them is an original vertex appearing as an in-neighbor only, and each pair of internal vertices is linked by a pair of arcs.
We say that every internal vertex added so far is a \emph{type-1} internal vertex.

However, at this point the in-degree of an original vertex can be much larger than three.
Thus, in order to bound the maximum in-degree of $D'$, in the second part of the construction we repeat the procedure for every original vertex $u$ with more than two in-neighbors, but this time using in-arborescences. 
Namely, we replace $u$ and every arc entering $u$ by a vertex-minimal in-arborescence with root $u$, maximum in-degree two, and leaves $N^-_{D'}(u$) (notice that this set contains only type-1 internal vertices). 
We replicate the steps of the first part of this construction accordingly but with one key observation.
We say that every internal vertex added in this part of this construction is a \emph{type-2} internal vertex and we do \textsl{not} add a pair of arcs in both directions between type-1 and type-2 internal vertices.
This implies that every arc between a type-1 and a type-2 is oriented from the former to the latter.
This property is important to bound the directed treewidth of expansions of DAGs (see \autoref{lem:directed-treewidth-expansion-of-dags}).
In the end, this construction generates a digraph with maximum degree seven.
This is witnessed, for example, by a type-1 internal vertex which is an out-neighbor of an original vertex $u$ with $\deg^+_{D'}(u) = 2$.

A crucial property is that, for every pair of original vertices $u,v \in V(D')$, there is a path from $u$ to $v$ in $D'$ if and only if there is a path from $u$ to $v$ in $D$.
This is true since there are no arcs from type-2 internal vertices to type-1 internal vertices.

\medskip
\noindent\textbf{Stars and paths.}
For an integer $\ell > 0$, an \emph{$\ell$-undirected star} is the undirected graph with $\ell+1$ vertices where every vertex other than the \emph{center} is linked to it by an edge.
In other words, an $\ell$-undirected star is a tree with $\ell+1$ vertices and $\ell$ \emph{leaves}.
We denote by $S^+_\ell$ (respectively $S^-_\ell$) the digraph obtained by orienting the all the edges of an $\ell$-star away from (respectively towards) the center.
We may also refer to $S^+_\ell$ and $S^-_\ell$ as an \emph{$\ell$-out-star} and an \emph{$\ell$-in-star}, respectively.
An $\ell$-star is an orientation of an $\ell$-undirected star.
An \emph{$\ell$-homogeneous star} is a digraph which is isomorphic to either $S^+_\ell$ or $S^-_\ell$.
We extend the definition of \emph{leaves} to the directed case.
We may omit the integer $\ell$ from these notations whenever the size of the star is unimportant.

Given an integer $k \geq 1$, we say that a digraph $H$ is a \emph{$k$-stars-paths} digraph if $H$ is formed by
\begin{itemize}
  \item the disjoint union of $k$ stars $S_1, \ldots, S_k$ with disjoint sets of leaves and  centers $C = \{s_i \in V(S_i) \mid i \in [k]\}$, and
  \item a collection of internally vertex disjoint paths $\mathcal{P}$ starting and ending in vertices of $C$ while avoiding the leaves of all stars.
\end{itemize}
We may also refer to $H$ by the triple $(\mathcal{S}, C, \mathcal{P})$, where $\mathcal{S} = \{S_1, \ldots, S_k\}$,  and say that the order $\order(H)$ of $H$ is $k$.
In \autoref{subsection:rooted_star_paths_subdigraphs} we show an \XP algorithm for a rooted version of {\sc Subdigraph Isomorphism} when the target is a stars-paths digraph.
See \autoref{fig:stars-paths-digraph} for an example of a $4$-stars-paths digraph.

\begin{figure}[h]
\centering
\scalebox{1}{\begin{tikzpicture}
\foreach \i in {1,3}{
	\node[vertex, scale=1] (u\i) at (3*\i,0) {$s_\i$};
	\foreach \j in {-2,...,0} {
		\node[blackvertex, scale=.4] (l\j) at ($(u\i) + (\j/2, -.75)$) {};
		\draw[arrow] (u\i) to (l\j);
	}
	\foreach \j in {1,2} {
		\node[blackvertex, scale=.4] (l\j) at ($(u\i) + (\j/2, -.75)$) {};
		\draw[arrow] (l\j) to (u\i);
	}
}

\foreach \i in {2,4}{
	\node[vertex, scale=1] (u\i) at (3*\i,0) {$s_\i$};
	\foreach \j in {-2,...,0} {
		\node[blackvertex, scale=.4] (l\j) at ($(u\i) + (\j/2, -.75)$) {};
		\draw[arrow] (l\j) to (u\i);
	}
	\foreach \j in {1,2} {
		\node[blackvertex, scale=.4] (l\j) at ($(u\i) + (\j/2, -.75)$) {};
		\draw[arrow] (u\i) to (l\j);
	}
}
\draw[arrow] (u1) to [bend left=20]
node [blackvertex, pos=1/8, scale=.4] {} 
node [blackvertex, pos=2/8, scale=.4] {} 
node [blackvertex, pos=3/8, scale=.4] {} 
node [blackvertex, pos=4/8, scale=.4] {} 
node [blackvertex, pos=5/8, scale=.4] {} 
node [blackvertex, pos=6/8, scale=.4] {} 
node [blackvertex, pos=7/8, scale=.4] {} 
(u2);

\draw[arrow] (u2) to [bend left=20]
node [blackvertex, pos=1/5, scale=.4] {} 
node [blackvertex, pos=2/5, scale=.4] {} 
node [blackvertex, pos=3/5, scale=.4] {} 
node [blackvertex, pos=4/5, scale=.4] {} 
(u3);

\draw[arrow] (u3) to [bend left=20]
node [blackvertex, pos=1/6, scale=.4] {} 
node [blackvertex, pos=2/6, scale=.4] {} 
node [blackvertex, pos=3/6, scale=.4] {} 
node [blackvertex, pos=4/6, scale=.4] {} 
node [blackvertex, pos=5/6, scale=.4] {} 
(u2);
\draw[arrow] (u4) to [bend right=20]
node [blackvertex, pos=1/3, scale=.4] {} 
node [blackvertex, pos=2/3, scale=.4] {} 
(u3);

\draw[arrow] (u4) to ($(u4) + (0,.5)$) to
node [blackvertex, pos=1/15, scale=.4] {} 
node [blackvertex, pos=2/15, scale=.4] {} 
node [blackvertex, pos=3/15, scale=.4] {} 
node [blackvertex, pos=4/15, scale=.4] {} 
node [blackvertex, pos=5/15, scale=.4] {} 
node [blackvertex, pos=6/15, scale=.4] {} 
node [blackvertex, pos=7/15, scale=.4] {} 
node [blackvertex, pos=8/15, scale=.4] {} 
node [blackvertex, pos=9/15, scale=.4] {} 
node [blackvertex, pos=10/15, scale=.4] {} 
node [blackvertex, pos=11/15, scale=.4] {} 
node [blackvertex, pos=12/15, scale=.4] {} 
node [blackvertex, pos=13/15, scale=.4] {} 
node [blackvertex, pos=14/15, scale=.4] {} 
 ($(u1) + (0, .5)$) to (u1);

\end{tikzpicture}}%
\caption{Example of a $4$-stars-paths digraph with centers $\{s_1, s_2, s_3, s_4\}$.}
\label{fig:stars-paths-digraph}
\end{figure}
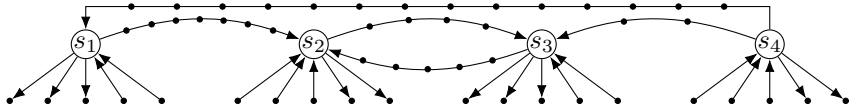

\medskip
\noindent\textbf{Caterpillars.}
An \emph{undirected caterpillar} is an undirected tree containing a path, called the \emph{spine}, using all non-leaf vertices.
We say that every vertex of degree at least three is a \emph{branching vertex} of the undirected caterpillar.
A (directed) caterpillar is an orientation of an undirected caterpillar such that the spine is a directed path and every branching vertex forms a homogeneous star with the leaves connected to it.

\subsection{Breakability}
A crucial notion in many algorithmic results using arboreal decompositions is that of the \emph{breakability} of a digraph $H$, which, to the best of our knowledge, had not been explicitly defined prior to our work. 
\begin{definition}[Breakability]\label{def:breakability}
Let $D$ be a digraph, $H \subseteq D$, and $w$ be a non-negative integer.
The \emph{$w$-breakability} of $H$, denoted by $\breakability_w(H)$, is the maximum number of weak components of $H[X]$ among all $w$-guarded sets of vertices  $X  \subseteq V(D)$.
We say that a class of digraphs $\mathcal{H}$ has \emph{bounded breakability} if there is a function $f(w)$ such that $\breakability_w(H') \leq f(w)$ for all $H' \in \mathcal{H}$.
\end{definition}
In other words, $\breakability_w(H)$ is one way to estimate how complex can be the interaction of $H$ with a $w$-guarded set when the goal is to obtain a dynamic programming algorithm using arboreal decompositions.
The reason is that in many algorithmic applications of arboreal decompositions (including ours and the ones by Johnson et al.~\cite{Johnson.Robertson.Seymour.Thomas.01} and Lopes and Sau~\cite{LopesS22}, for example), in order to exploit the property of the guards, it is necessary to guess how the digraph associated with a solution can enter and leave a $w$-guarded set.
Thus, if a digraph associated with a solution of a given problem $\Pi$ has unbounded breakability, we can understand this as an indication that a dynamic programming approach for $\Pi$ relying on arboreal decompositions is unlikely to be viable.
We ask whether it is reasonable to expect such algorithms when the breakability is bounded.
Sadly, the answer is no.
In \autoref{sec:hardness_results}, among other results, we show that {\sc Subdigraph Isomorphism} remains \NP-complete when the target digraph has both bounded breakability and maximum degree.

In \autoref{sec:algorithm}, we use the bound on the breakability of paths, which was proven by Johnson et al.~\cite{Johnson.Robertson.Seymour.Thomas.01}.
The proof follows almost immediately from \autoref{def:guarded-sets}: if $A$ is a $w$-guarded set, then any subpath of a path $P$ starting and ending in $A$ and intersecting $V(D) \setminus A$ must be intersected by a guard of $A$. 
\begin{proposition}[Johnson et al.~{\cite[4.5]{Johnson.Robertson.Seymour.Thomas.01}}]
\label{proposition:paths-breakability}
Let $D$ be a digraph formed by the disjoint union of $k \geq 1$ paths and $w$ be a non-negative integer.
Then $\breakability_w(D) \leq k+w$. 
\end{proposition}

\section{\texorpdfstring{\XP}{XP} Algorithm for stars-paths subdigraphs}\label{sec:algorithm}
In this section we show our \XP algorithm for \textsc{\pname}, formally defined as follows.

%\vspace{-.55cm}
\begin{boxproblem}[framed]{\pname (\pshort)}
\label{problem:main_problem_rooted_stars}
Input: & A digraph $D$, a stars-paths digraph $H = (\mathcal{S}, C, \mathcal{P})$  with $\mathcal{S} = \{S_1, \ldots, S_k\}$ and $C = \{s_i \in V(S_i) \mid i \in [k]\}$, a set $X \subseteq V(D)$, and a bijective mapping $f: C \to X$.\\
Output: & A subdigraph $H' \subseteq D$ such that $H' \cong H$ and\\
  & \lipItem{1.} for all $i \in [k]$, $f(s_i)$ is the center of a star $S'_i \subseteq H'$ with $S'_i \cong S_i$; and\\
  & \lipItem{2.} for all $P \in \mathcal{P}$ from $s_i$ to (not necessarily distinct) $s_j$ in $H$, there is a $P' \subseteq D$ such that $|V(P')| = |V(P)|$ and $P'$ is a path from $f(s_i)$ to $f(s_j)$ in $D$.
\end{boxproblem}

Notice that $H' \cong H$ implies that the paths mentioned in item \lipItem{2.} above are pairwise internally vertex disjoint.

Mainly, the algorithm is divided in two parts. 
In the first part of the algorithm, we show how to decide if the stars forming $H$ can be embedded in $D$ while respecting $f$.
We achieve this through  a system of inequations that tests if the leaves of the stars can be taken from the neighborhoods of the vertices in $\{f(v) \mid v \in C\}$.
Moreover, we add another ingredient to this system which is required in the second part: not only we test if the stars can be built, but also how many unused vertices we can keep from each $N_D(f(v))$ with $v \in C$.
We refer to this number as the \emph{slack} of $v$ and remark that, depending on how we choose the leaves for the stars in $D$, many different combinations of slacks are possible and all of them must be considered.
This part of the algorithm is \FPT in $k$ assuming that the slacks are given (since we are solving the rooted version of the problem and thus there is no guessing on where to find the centers of the stars in $D$) and is shown in \autoref{subsection:finding_the_stars_and_the_slacks}.

In possession of the slacks, we proceed to the second part of the algorithm.
Here we introduce a generalization of the classical \textsc{Directed Disjoint Paths} problem, which we call \textsc{\pddp} (\textsc{\pddpshort} for short) which is used to find a collection of pairwise disjoint paths linking the terminals such that the union of those paths avoids at least a given number $x_{f(v)}$ of vertices in the neighborhood of $f(v)$, for each $v \in C$.
For each $f(v)$, we include in the input of this problem a set $X_{f(v)}$ from the neighborhood of $f(v)$ in $D$ and a number $x_{f(v)}$ of vertices from this set that we want to avoid using when building the paths.
Our algorithm runs in \XP time parameterized by the number of paths, the number of subsets, and the directed treewidth of the input digraph.
We build on the framework by Johnson et al.~\cite{Johnson.Robertson.Seymour.Thomas.01} for \textsc{DDP} to \textsc{\pddpshort} to extend their result to \textsc{Directed Disjoint Paths} to our case (see also~\cite{LopesS22}).

In \autoref{subsection:rooted_star_paths_subdigraphs} we use those two ingredient to solve \textsc{\pname}.
In short, for each possible choice of the slacks we define a system of inequations for the stars in the first part.
For each system with a positive solution, we use the slacks to generate an instance of \textsc{\pddpshort} and solve it in \XP time.
If any of those instances is positive, we have found our desired subdigraph of $D$ and the algorithm terminates.
Otherwise, we correctly conclude that the given instance of \textsc{\pshort} is negative.
The \XP time comes not only from our algorithm for \textsc{\pddpshort}, but also from the fact that \textsl{all} possible choices of the slacks must be considered, and there are $n^{\Ocal(k)}$ possible choices.

\subsection{Finding the stars and the slacks}\label{subsection:finding_the_stars_and_the_slacks}
The goal of this section is to solve the following auxiliary problem.
\begin{boxproblem}[framed]{Rooted Stars Subdigraph Isomorphism}
\label{problem:rooted_stars_subdigraph_isomorphism}
Input: & A host digraph $D$, a digraph $H$ formed by the disjoint union of $k$ stars $\{S_1, \ldots, S_k\}$ with centers $C = \{s_1, \ldots, s_k\}$, a set $X \subseteq V(D)$, and a mapping $f: C \to X$.\\
Output: & A subdigraph $H' \subseteq D$ such that $H' \cong H$ and for all $i \in [k]$ the vertex $f(s_i)$ is the center of a star $S'_i \subseteq D$ with $S'_i \cong S_i$.
\end{boxproblem}
Notice that in this case we no longer require $f$ to be bijective.
Given an instance of this problem, we build a integer system of $f(k)$ linear inequalities whose solution is used to find the desired subdigraph $H' \subseteq D$.
Then, we formally define the \emph{slacks} which are needed for the second part of the algorithm and show how to add further inequalities to this system to obtain every feasible combination of the slacks.

\begin{theorem}\label{theorem:rooted-stars-isomorphism-is-FPT}
{\sc Rooted Stars Subdigraph Isomorphism} is \FPT with parameter $k$.
\end{theorem}
\begin{proof}
Let $D$, $H$, $X$, $C$, $f$, and the stars $\{S_1, \ldots, S_k\}$ be given as in the definition of the problem.
First, notice that by increasing the number of stars of $H$ to at most $2k$ we can assume that every star in $\mathcal{S}$ is homogeneous.
Indeed, if, for instance, some $S_i$ is not homogeneous, we make a copy of center $s_i$ of $S_i$ and, instead of $S_i$, we search for one out-star with as many leaves as $\deg^+_{S_i}(s_i)$ and for one in-star with as many leaves as $\deg^-_{S_i}(s_i)$.
Then we extend $f$ to map the centers of those two new stars to the vertex $f(s_i)$, and remove $s_i$ from the domain of $f$.

Assume now that every $S_i$ is a homogeneous star.
For $i \in [k]$, let $s'_i = f(s_i)$, let $\ell_i$ be the number of leaves of $S_i$, and 
\[ N_i = 
\begin{cases}
N^-_D(s'_i), & \text{if } S_i \text{ is an in-star},\\
N^+_D(s'_i), & \text{if } S_i \text{ is an out-star}.\\
\end{cases}
\]
Notice that $s'_i \in V(D)$.
For $J \in 2^{[k]} \setminus \{\emptyset\}$, let $X_J = \bigcap_{j \in J} N_J$.
In other words, $X_J$ contains vertices which are neighbors of \emph{exactly} the vertices $s'_j$, for $j \in J$, and that are candidates for the leaves of stars $S'_j$. Note that $\{X_J \mid J \in 2^{[k]}\}$ partitions $\bigcup_{i \in [k]} N_i$.
Finally, for for all $i \in [k]$ and $J \in 2^{[k]}\setminus \{\emptyset\}$ we create an integer variable $t_i(J)$ corresponding to the number of vertices taken from $X_J$ for the leaves of the star $S'_i \subseteq D$ we want to build.

From the point of view of the stars $S'_i$, the goal is that $\ell_i$ vertices of $N_i$ are taken for its leaves.
Thus, for a fixed choice of $i$, we want that
\[\sum_{J : i \in J} t_i(J) = \ell_i.\]
The restriction is simply that each star can only take from $X_J$ as many vertices as there are in the set. 
For a fixed choice of $J$, we write this as
\[\sum_{i \in J}t_i(J) \leq |X_J|.\]
Adding all together, we get the following system of $2^k - 1$ variables:
\begin{align}\label{eq:first-system-no-slacks}
\sum_{J : i \in J} t_i(J) = \ell_i, & \hfill & \forall i \in [k].\\
\sum_{i \in J}t_i(J) \leq |X_J|, & \hfill & \forall J \in 2^{[k]} \setminus \{\emptyset\}.
\end{align}
A solution for this system gives us viable choices for the values of $t_i(J)$ which, in turn, tells us how many vertices of $X_J$ should be taken for the leaves of $S'_i$, and (1) above ensures that $\ell_i$ vertices can be taken.
Since $\{X_J \mid J \in 2^{[k]}\}$ partitions $\bigcup_{i \in [k]} N_i$, no vertex is taken more than once.
Since an integer system of linear inequations can be solved in \FPT time parameterized by the number of variables~\cite{Lenstra83}, the result follows.
\end{proof}

For $i \in [k]$, we denote by $\slack(i)$ the number of vertices in the in- or out-neighborhood of $f(s_i)$, depending on whether $S_i$ is an in- or out-star (we can assume every star is homogeneous as in the proof of \autoref{theorem:rooted-stars-isomorphism-is-FPT}), that we want to \textsl{avoid} using as leaves of the stars.
Informally, the goal is to provide some leeway in how the paths are constructed in the second part of the algorithm for \textsc{\pname}: if the paths are constructed in such way that they avoid at least $\slack(i)$ vertices in the neighborhood of $f(s_i)$ and a solution for the system in the proof of \autoref{theorem:rooted-stars-isomorphism-is-FPT} respecting the slacks exists, then one can both route the paths and build the stars to find the desired stars-paths subdigraph.
In what follows, we formalize this intuition.

If $\slack(i)$ is given for $i \in [k]$, we can test whether we can build the stars while respecting the slacks in an instance of \textsc{Rooted Stars Subdigraph Isomorphism} by adding one more set of $k$ inequations to the system formed by (1) and (2) above.
Adopting the notation used in the proof, we want that
\begin{align}\label{eq:slacks-inequation}
\sum_{J : i \in J}\left(\sum_{a \in J}t_a(J)\right) \leq |N_i| - \slack(i), & \hfill & \forall i \in [k].
\end{align}
The interpretation of \autoref{eq:slacks-inequation} is as follows.
For all $i \in [k]$ we need to ensure that at most $|N_i| - \slack(i)$ are taken from $N_i$.
To write this using the variables of the form $t_a(J)$, we look at every set $X_J$ with $i \in J$ (the outermost summation) and then count how many vertices of this set were taken for the leaves of stars that can actually use those vertices (the innermost summation).
We say that an instance of {\sc Rooted Stars Subdigraph Isomorphism} \emph{respects} the slacks if there is a solution for the instance satisfying \autoref{eq:slacks-inequation}.
From this analysis, we obtain the following corollary.
\begin{corollary}\label{corollary:rooted-stars-subdigraph-and-slacks}
Given an instance $\mathcal{I}$ of {\sc Rooted Stars Subdigraph Isomorphism}  and $\slack(i)$ for every $i \in [k]$, one can decide if there is a solution for $\mathcal{I}$ respecting the slacks in \FPT time parameterized by $k$.
\end{corollary}

\subsection{Routing the paths}\label{subsection:routing_the_paths}
In this section we formally define and solve the \textsc{\pddp} problem.
For convenience, we adopt the following notation.
\begin{definition}[Requests]
\label{definition:requests}
Let $D$ be a digraph.
A \emph{request} in $D$ is an ordered pair of vertices of $D$.
If $R$ is a request set $\{(s_1, t_1), \ldots, (s_k, t_k)\}$ and $\mathcal{P}$ is a collection of paths $\{P_1, \ldots, P_k\}$ such that, for $i \in [k]$, $P_i$ is a path from $s_i$ to $t_i$ in $D$, we say that $\mathcal{P}$ \emph{satisfies} $R$.
\end{definition}
We say that a request set $R$ in $D$ is \emph{contained} in set $A \subseteq V(D)$ if every vertex occurring in a pair of $R$ is in $A$. 
We also adopt standard notation for arrays of integers: if $\xarray \in [n]^k$, for $i \in [k]$ we denote by $\xarray[i]$ the value of the $i$-th index of $\xarray$.

\begin{boxproblem}[framed]{\pddp (\pddpshort)}
\label{problem:subset_ddp}
Input: & A digraph $D$, a request set $R$ of size $r \geq 1$, subsets $X_1, \ldots, X_k$ of $V(D)$, non-negative integers $x_1, \ldots, x_k$ and $p_1, \ldots, p_r$\\
Output: & A set of pairwise internally disjoint paths $\mathcal{P} =  \{P_1, \ldots, P_k\}$ satisfying $R$ such that the union of the paths in $\mathcal{P}$ contains at most $x_i$ vertices of $X_i$ for $i \in [k]$, and $|V(P_i)| = p_i$ for $i \in [r]$.
\end{boxproblem}
\noindent For the remainder of this subsection, we assume that $\mathcal{I}$ is an instance of \textsc{\pddpshort} with input digraph $D$, request set $R = \{(s_1, t_1), \ldots, (s_k, t_k)\}$, subsets $X_1, \ldots, X_k \subseteq V(D)$, and non-negative integers $x_1, \ldots, x_k$.
We denote by $w$ the directed treewidth of $D$ and, for convenience, we may define instances of \textsc{\pddpshort} using arrays of integers instead of explicitly writing $x_1, \ldots, x_k$ and $p_1, \ldots, p_k$.
The sizes of the paths is important since in the goal is to find a subdigraph of a given digraph.

We apply the framework by Johnson et al.\cite[Section 4]{Johnson.Robertson.Seymour.Thomas.01} to provide an \XP algorithm with parameters $r$, $k$, and $\dtw(D)$ for \textsc{\pddpshort}.
Following their notation, in our dynamic programming scheme we refer to the information that we want to compute at each step of the algorithm as the \emph{itinerary}, which we formally define later.
Thus, the goal is to compute an itinerary for $V(D)$.
The framework introduced in~\cite{Johnson.Robertson.Seymour.Thomas.01} gives two conditions that, together, suffice to provide an \XP algorithm for the problem.

\begin{condition}[Johnson et al.~{\cite[Axiom 1]{Johnson.Robertson.Seymour.Thomas.01}}]\label{condition:big-big}
Let $D$ be a digraph, $w$ and $\alpha$ be non-negative integers, and $A,B$ be two disjoint $w$-guarded subsets of $V(D)$ such that there are no arcs from $B$ to $A$ in $D$.
Then an itinerary for $A \cup B$ can be computed from an itinerary for $A$ and an itinerary for $B$ in time $\Ocal(n^{\alpha})$.
\end{condition}

\begin{condition}[Johnson et al.~{\cite[Axiom 2]{Johnson.Robertson.Seymour.Thomas.01}}]
\label{condition:big-small}
Let $D$ be a digraph, $w$ and $\alpha$ be a non-negative integers, and $A,B$ be two disjoint subsets of $V(D)$ such that $A$ is $w$-guarded and $|B| \leq w$.
Then an itinerary for $A \cup B$ can be computed from an itinerary for $A$ in time $\Ocal(n^{\alpha})$.
\end{condition}

\begin{theorem}[Johnson et al.~{\cite[4.4]{Johnson.Robertson.Seymour.Thomas.01}}]\label{theorem:algo-from-conditions}
Provided that \cref{condition:big-big,condition:big-small} hold, there is an algorithm running in time $\Ocal(n^{\alpha+1})$ that receives as input a digraph $D$ and an arboreal decomposition of $D$ with width at most $w$ and outputs an itinerary for $V(D)$.
\end{theorem}

The strategy is to exploit the breakability of paths to show that \cref{condition:big-small,condition:big-big} are satisfied.
We remind the reader of \autoref{proposition:paths-breakability}.
In short, it implies that, when given a $w$-guarded set $A$, in our algorithm we can pay the cost of guessing how the paths of a solution $\mathcal{P}$ for $\mathcal{I}$ interact with $A$.
We now formally define an itinerary for \textsc{\pddpshort}. 

\begin{definition}[Itinerary]
For $A \subseteq V(D)$ let $\mathcal{R}_A$ be the set of all request sets of size at most $k+w$ on $D$ which are contained in $A$.
An \emph{$(\mathcal{I},w)$-itinerary} for $A$ is a function $f_A: \mathcal{R}_A \times [n]^{k} \times [n]^k \to \{0,1\}$ such that $f_A(L, \xarray, \parray) = 1$ if and only if
\begin{enumerate}
  \item for $i \in [k]$ it holds that $\xarray[i] \leq x_i$; and
  \item the instance $(D[A], L, \{X_1 \cap A, \ldots, X_k \cap A\}, \{\xarray[i] \mid i \in [k]\}, \{\parray[i] \mid i \in [k])\}$ of \textsc{\pddpshort} is positive.
\end{enumerate}
\end{definition}

In other words, $f_A(L, \xarray, \parray) = 1$ implies that, in $D[A]$, we can find a set of paths satisfying the request set $L$ avoiding at least $\xarray[i]$ vertices of $X_i \cap A$ for all $i \in [k]$ and with the desired sizes of the paths.
Thus $\mathcal{I}$ is positive if and only if $f_{V(D)}(R, \xarray', \parray') = 1$ with $\xarray'[i] = x_i$ and $\parray'[i] = p_i$ for $i \in [k]$.
We remark that we do not change the sets $X_1, \ldots, X_k$ throughout the algorithm.
Indeed this is needed as enumerating all possible subsets of $X_i$ that we can try to avoid when constructing the paths inside of $A$ is too costly.
Thus, we maintain these sets and, for each $A$, respond the following question: given a request set contained in $A$, is it possible to find the paths satisfying the requests while avoiding $\xarray[i]$ vertices of $X_i \cap A$?
Computing this answer is within our means since there are $\Ocal(n^{k+w})$ request sets contained in $A$ and $\Ocal(n^k)$ choices for $\xarray$ and $\parray$.
We are now ready to prove that \cref{condition:big-big,condition:big-small} hold for \textsc{\pddpshort}.

\begin{lemma}\label{lemma:condition_1_holds}
Let $A,B$ be disjoint and $w$-guarded subsets of $V(D)$ such that there no arcs from $B$ to $A$ in $D$.
Then an $(\mathcal{I}, w)$-itinerary for $A \cup B$ can be computed from itineraries for $A$ and $B$ in time $\Ocal(n^{4(r+w)+3k+2})$.
\end{lemma}
\begin{proof}
Let $f_A$ and $f_B$ be $(\mathcal{I}, w)$-itineraries for $A$ and $B$, respectively, and let $L = \{(s_1, t_1)$, $\ldots$, $(s_\ell, t_\ell)\}$ be a request set contained in $A \cup B$ with $\ell \leq r+w$.
Let $\xarray \in n^{[k]}$ and $\parray \in n^{[\ell]}$. 
The goal is to compute $f_{A \cup B}(L, \xarray, \parray)$ using $f_A$ and $f_B$.

We set $f_{A \cup B}(L, \xarray, \parray) = f_A(L, \xarray, \parray)$ if $L$ is contained in $A$ and $f_{A \cup B}(L, \xarray, \parray) = f_B(L, \xarray, \parray)$ if $L$ is contained in $B$.
If there is $(s,t) \in L$ such that $t \in A$ and $s \in L$ then we set $f_{A \cup B}(L, \xarray, \parray) = 0$ since there are no arcs from $B$ to $A$ in $D$.
Assume now that $L$ is not contained in $A$ nor in $B$ and that no such pair $(s,t)$ with $s \in B$ and $t \in A$ exists in $L$.

The strategy is as follows.
Whenever we are faced with a request $(s,t) \in L$ with $s \in A$ and $t \in B$, we add ``virtual'' terminals to both $A$ and $B$ in order to extract the value of $f_{A \cup B}(L, \xarray)$ from isolated instances in $A$ and $B$ which can be verified using the itineraries for $A$ and $B$.
Formally, start with $L_A = L_B = \emptyset$.
For $i \in [\ell]$, do the following.

\begin{enumerate}
  \item If $s_i \in A$ and $t_i \in A$, define $s^A_i = s_i$ and $t^A_i = t_i$, and include the request $(s^A_i, t^A_i)$ in $L_A$.
  \item If $s_i \in B$ and $t_i \in B$, define $s^B_i = s_i$ and $t^B_i = t_i$, and include the request $(s^B_i, t^B_i)$ in $L_B$.
  \item If $s_i \in A$ and $t_i \in B$, define $s^A_i = s_i$ and $t^B_i = t_i$. Then, choose $t^A_i \in A$ and $s^B_i \in B$ such that there is an arc from $t^A_i$ to $s^B_i$ in $D$ and include $(s^A_i, t^A_i)$ in $L_A$ and $(s^B_i, T^B_i)$ in $L_B$. 
\end{enumerate}

See \autoref{fig:construction_of_lemma_condition_1_holds} for an illustration of this construction.
\begin{figure}[h]
\centering
\scalebox{1}{\begin{tikzpicture}

\node[blackvertex, label=180:{$s^A_1$}, scale=.6] (sa1) at (0,2) {};
\node[blackvertex, label=180:{$s^A_2$}, scale=.6] (sa2) at (0,1) {};
\node[blackvertex, label=180:{$s^A_3$}, scale=.6] (sa3) at (0,0) {};

\foreach \i in {1,2,3} {
	\node[blackvertex, label = 90:{$t^A_\i$}, scale=.6] (ta\i) at ($(sa\i) + (1.7,0)$) {};
	\draw[arrow,snake it] (sa\i) to (ta\i);
	\node[blackvertex, label = 90:{$s^B_\i$}, scale=.6] (sb\i) at ($(ta\i) + (1.7,0)$) {};
	\node[blackvertex, label = 0:{$t^B_\i$}, scale=.6] (tb\i) at ($(sb\i) + (1.7,0)$) {};
	\draw[arrow,snake it] (sb\i) to (tb\i);
}
\draw[arrow] (ta2) to (sb2);

\node[draw, rectangle, fit=(sa1)(ta3), yscale=1.3, xscale=1.3, yshift = .1cm, xshift = -.1cm, label=$A$] {};
\node[draw, rectangle, fit=(sb1)(tb3), yscale=1.3, xscale=1.3, yshift = .1cm, xshift = .1cm, label=$B$] {};

\end{tikzpicture}}
\caption{Construction of the requests sets of \autoref{lemma:condition_1_holds}. In the example, terminals $t^A_2$ and $s^B_2$ were added as in item \lipItem{3.} for the original request $(s_2,t_2) \in L$ with $s^A_2 = s_2$ and $t^B_2 = t_2$.}
\label{fig:construction_of_lemma_condition_1_holds}
\end{figure}
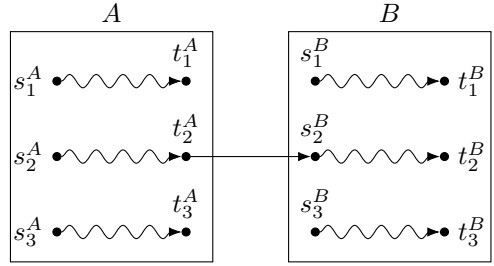

Now we set $f_{A\cup B}(L, \xarray, \parray) = 1$ if there are $\xarray^A$, $\xarray^B$, and $\parray^A, \parray^B$ such that $f_A(L_a, \xarray^A, \parray^A) = f_B(L_B, \xarray^B, \parray^B) = 1$, $\xarray^A[i] + \xarray^B[i] \leq \xarray[i]$ for all $i \in [k]$, and $\parray^A[i] + \parray^B[i] = \parray[i]$ for all $(s_i, t_i)\in L$ where $s_i \in A$ and $t_i \in B$. 
If this is not the case, we repeat the procedure to construct $L_A$ and $L_B$ above with different choices fo $t^A_i$ and/or $s^B_i$ in the third step.
If all possible choices have been exhausted without a positive answer, we set $f_{A\cup B}(L, \xarray, \parray) = 0$.
We claim that this assignment is correct.

Consider the instance $\mathcal{I}' = (D[A \cup B], L, \{X_i \cap (A \cup B) \mid i \in [k]\}, \xarray, \parray)$ of \textsc{\pddpshort}.
If it is positive then some of the paths of a solution $\mathcal{P}$ satisfying $L$ are contained in $A$, some are contained in $B$, and some are crossing from $A$ to $B$.
We refer to these last paths as \emph{crossing} paths.
We denote by $\mathcal{P}_A$ the set containing all paths of $\mathcal{P}$ contained in $A$ plus the subpaths of the crossing paths which are contained in $A$.
Similarly, we denote by $\mathcal{P}_B$ the set containing all paths of $\mathcal{P}$ contained in $B$ plus the subpaths of the crossing paths which are contained in $B$.
This implies that for some choice of $L_A$ and $L_B$ the collections $\mathcal{P}_A$ and $\mathcal{P}_B$ are solutions for the instances $(D[A], L_A, \{X_i \cap A \mid i \in [k]\}, \xarray^A, \parray^A)$ and $(D[B], L_B, \{X_i \cap B \mid i \in [k]\}, \xarray^B, \parray^B)$ and $\xarray^A[i] + \xarray^B[i] \leq \xarray[i]$ for all $i \in [k]$.
Hence the definition of itineraries implies that $f_A(L_A, \xarray^A, \parray^A) = f_B(L_B, \xarray^B, \parray^B) = 1$.

Conversely, if $f_A(L_A, \xarray^A, \parray^A) = f_B(L_B, \xarray^B, \parray^B) = 1$ for $L_A, L_B$ constructed following the steps described above, one can build a collection of paths satisfying $L$ by gluing together the paths satisfying $L_A$ and the paths satisfying $L_B$.
If $\xarray^A[i] + \xarray^B[i] \leq \xarray[i]$ for all $i \in [k]$ then $\mathcal{I}'$ is positive and the claim follows.

The running time follows by observing that there are $\Ocal(n^{(2(r+w)})$ choices for the request set $L$, $\Ocal(n^{2(r+w)})$ choices for the pairs $(t^A_i, s^B_i)$ in item \lipItem{3.} above, $\Ocal(n^k)$ choices for each of the arrays $\xarray, \xarray^A$, and $\xarray^B$, and $\Ocal(n^\ell)$ choices for $\parray^A, \parray^B$.
\end{proof}

\begin{lemma}\label{lemma:condition_2_holds}
Let $A,B$ be disjoint subsets of $V(D)$ such that $A$ is $w$-guarded and $|B| \leq w$.
Then an $(\mathcal{I}, w)$-itinerary for $A \cup B$ can be computed from itineraries for $A$ and $B$ in time $\Ocal(n^{4(r+w)+3k+2})$.
\end{lemma}
\begin{proof}
Let $f_A$ be a $(\mathcal{I}, w)$-itinerary for $A$ and $L$ be a request set contained in $A \cup B$ with $L \{(s_1, t_1), \ldots, (s_\ell, t_{\ell})\}$ and $\ell \leq r+w$.
Let $\xarray \in [n]^k$, let $\parray \in n{[\ell]}$, and $\mathcal{I'}$ be the instance $(D[A \cup B], L, \{X_i \cap (A \cup B) \mid i \in [k]\}, \xarray, \parray)$.
In this case, there is no restriction on the arcs between $A$ and $B$.
However, the bound on the size of $B$ allows us to brute-force all the ways that paths in a solution for $\mathcal{I}'$ can interact with $B$.

For each pair $(s,t) \in L$ a path from $s$ to $t$ in $D[A \cup B]$ can \lipItem{(i)} be entirely contained in $A$, \lipItem{(ii)} be entirely contained in $B$, or \lipItem{(iii)} intersect both $A$ and $B$.
If $L$ is contained in $A$, we can test if there is a solution for $\mathcal{I}'$ whose paths are all contained in $A$ by verifying the value of $f_A(L, \xarray^A, \parray^A)$ for every possible choices of $\xarray^A$ and $\parray^A$.
If $L$ is contained in $B$, we can test if there is a solution entirely contained in $B$ in time $\Ocal(2^{w} \cdot n^k)$.
We now search for solutions containing paths intersecting both $A$ and $B$.
Notice that this does not imply that $L$ is not contained in $A$ or in $B$: even if it is, for instance, contained in $B$, it might be the case that a path needs to traverse both $A$ and $B$ to satisfy a request.

The first step is to understand how a path from a solution can be broken into pieces contained in $A$ and $B$.
Suppose that $P$ is a path from a solution $\mathcal{P}$ of $\mathcal{I}'$ intersecting both $A$ and $B$.
Let $\mathcal{P}_A = \{P^A_1, \ldots, P^A_a\}$ be the set of subpaths of $P$ which are contained in $A$ and, for $i \in [a]$, let $(u_i, v_i)$ be the pair containing the first and last vertices occurring in $P^A_i$.
Now, let $\mathcal{P}_B$ be the collection of subpaths of $P$ contained in $B \cup \{u_i,v_i \mid i \in [a]\}$.
For this analysis, we distinguish two cases.
If $a = 1$, then $\mathcal{P}_B$ is a (possibly unitary) collection of pairwise disjoint paths satisfying the request $(s, u_1)$ if $s \in B$ and the request $(v_1, t)$ if $t \in B$.
If $a \geq 2$, then $\mathcal{B}$ is a (again possibly unitary) collection of pairwise disjoint paths satisfying the requests $\{(v_1, u_2), \ldots, (v_{a-1}, u_a)\}$ together with $(s, u_1)$ if $s \in B$ and $(v_1, t)$ if $t \in B$.
In any case, the crucial message is that if we guess which vertices of a path in a solution are to be linked through $B$, we can generate a request set from this guess and compute if its possible to find such paths through $B$.
Moreover, since $A$ is $w$-guarded, \autoref{proposition:paths-breakability} tell us that only a reasonable amount of guesses need to be processed.
See \autoref{fig:construction_of_lemma_condition_2_holds} for an illustration of this analysis.
\begin{figure}[h]
\centering
\scalebox{1}{\begin{tikzpicture}[yscale=.8]
\foreach \i in {1,2,3} {
	\node[blackvertex, scale=.6, label=180:{$u_\i$}] (u\i) at (4*\i, 0) {};
	\node[blackvertex, scale=.6, label=0:{$v_\i$}] (v\i) at ($(u\i) + (2,0)$) {};
	\draw[arrow, goodblue] (u\i) to [bend left = 20] (v\i);
}

\node[blackvertex, scale=.6, label=180:{$s$}] (s) at (3,-1.5) {};
\node[blackvertex, scale=.6, label=0:{$t$}] (t) at (15,-1.5) {};

\foreach \i/\j in {1/2,2/3} {
	\draw[arrow, goodred] (v\i) to [out=-90, in=180] ($(v\i) + (1,-1.5)$) to [out=0, in=-90] (u\j);
}
\draw[arrow, goodred] (s) to [out=90, in=180] ($(s)+ (.5, 1)$) to [out=0,in=-90] (u1);
\draw[arrow, goodred] (v3) to [out=-90, in=180] ($(v3)+ (.5, -1)$) to [out=0,in=90] (t);

\node[label=above:$A$, label=below:$B$] (A) at ($(s) + (-.5, 1.1)$) {};
\node (B) at ($(t) + (.5, 1.1)$) {};
\draw[dashed] (A) -- (B);
\end{tikzpicture}}%
\caption{Illustration of how a path in a solution can be split into subpaths in $A$ and $B$. In the example, $s,t \in B$ and the request set satisfied by $\mathcal{P}_B$ is $\{(s, u_1), (v_1, u_2), (v_2, u_3), (v_3, t)\}$. The paths in $\mathcal{P}_A$ are depicted in blue, and the paths in $\mathcal{P}_B$ in red.}
\label{fig:construction_of_lemma_condition_2_holds}
\end{figure}
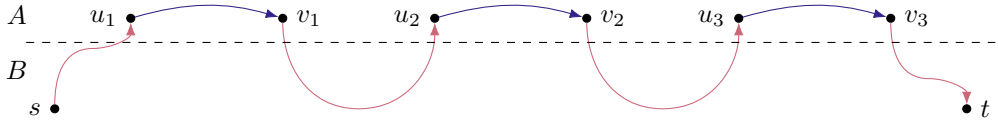
Thus we can test whether there is a solution for $\mathcal{I}'$ using an itinerary for $A$ and, for each $(s_i,t_i) \in L$, searching for a collection $\mathcal{P}_B$ as described above.
For each $(s_i,t_i) \in L$, we guess which parts of a path satisfying $(s_i, t_i)$ are in $A$ (the paths from $u_i$ to $v_i$ in the analysis above) and link the missing pieces (the paths from $v_{i}$ to $u_{i+1}$ in the analysis above) through $B$.
By \autoref{proposition:paths-breakability}, a path has at most $r+w$ subpaths in $A$ and thus at most $\Ocal(n^{r+w})$ requests sets $L_A$ contained in $A$ need to be considered.
For each of those requests, we can test if we can link the missing pieces through $B$ in time $\Ocal(2^r)$.
Finally, we set $f_{A \cup B}(L, \xarray, \parray) = 1$ if and only if $f_A(L_A, \xarray^A, \parray^A) = 1$ and we can link the missing pieces through $B$ in such way that the paths routed through $B$ use at most $\xarray^B[i]$ vertices of $X_i \cap B$, for all $i \in [k]$, for a choice of $\xarray^B$ such that $\xarray^A[i] + \xarray^B[i] \leq \xarray[i]$, and satisfying the requirements for the sizes of the paths, and the result follows.
\end{proof}

The main theorem of this section follows immediately from \cref{lemma:condition_1_holds,lemma:condition_2_holds} and \autoref{theorem:algo-from-conditions}.

\begin{theorem}\label{theorem:algo-for-the-paths}
The {\sc \pddp} problem is solvable in time $\Ocal(n^{4(r+w)+3k+3})$.
\end{theorem}

\subsection{Rooted stars-paths subdigraphs}\label{subsection:rooted_star_paths_subdigraphs}
In this section, we add together the ingredients from \cref{subsection:rooted_star_paths_subdigraphs,subsection:routing_the_paths} to obtain an algorithm for \textsc{\pname}.
Let $\mathcal{I} = (D, H, X, f)$ be an instance of \textsc{\pshort} as in the definition of the problem.
Thus $H = (\mathcal{S}, C, \mathcal{P})$ is built from stars $\mathcal{S} = \{S_1, \ldots, S_k\}$ with centers $C = \{s_i \in V(S_i) \mid i \in [k]\}$, and $f$ is a bijective mapping from $C$ to a set $X \subseteq V(D)$ dictating where the centers of the stars must be in $D$.

First, we determine the feasible choices for the slacks.
For this, we filter $\mathcal{I}$ into an instance $\mathcal{I}' = (D, H', C', X, f')$  of \textsc{Rooted Stars Subdigraph Isomorphism} built as follows.
We maintain $D$ and $X$ unaltered.
Then, we change $H'$ and $C'$ in order to ensure that all stars forming $H'$ are homogeneous.
This is done as in the proof of \autoref{theorem:rooted-stars-isomorphism-is-FPT}.
Namely, if a given star $S$ with center $s$ forming $H$ is not homogeneous, we delete $S$ from $H$ and add to it two stars $S^+$ and $S^-$ with as many leaves as $\deg^+_H(s)$ and $\deg^-_H(s)$, respectively, and update the centers to $C'$ adequately and $f$ into $f'$ to map the centers of the two new stars to $f(s)$ (we remind the reader that the mapping $f'$ need not be bijective).
In the end, we are left with $k'$ homogeneous stars $\{S'_1, \ldots, S'_k\}$ in $H'$ with centers $C' = \{s'_i \in V(S'_i) \mid i \in [k'] \}$ and $k' \leq 2k$, and the mapping $f': C' \to X$.

For each $\xarray \in [n]^{k'}$, we apply \autoref{corollary:rooted-stars-subdigraph-and-slacks} with input $\mathcal{I'}$ and $\slack(i) = \sarray[i]$ for all $i \in k'$.
If we obtain a solution for $\mathcal{I'}$ respecting the slacks, we proceed to the routing step, where the goal is to find the paths in $D$ between the centers of the stars in order to build the desired $H'$.
The instance of \textsc{\pddpshort} is built as follows.
For each $i \in [k']$ we let $X_i = N^+_D(f'(s'_i))$ if $S'_i$ is an out-star and $X_i = N^-_D(f'(s'_i))$ if $S'_i$ is an in-star, and define $x_i = \xarray[i]$.
The input digraph $D$ is maintained and, to build the request set, we look at the collection of paths $\mathcal{P}$ which is a defining part of $H$.
Let $\mathcal{P} = \{P_1, \ldots, P_r\}$ and for $j \in [r]$ let $s_i$ and $t_i$ be the first and last vertices of $P_i$.
We define $R = \{(f'(s_i), f'(t_i) \mid i \in [r]\}$ and solve the instance $(D, R, \{X_1, \ldots, X_k\}, (x_1, \ldots, x_k))$ of \textsc{\pddpshort}.
If the answer is positive, we have found the paths between the centers of the stars avoiding sufficiently many vertices of the relevant neighborhoods of the vertices in $X$ such that the stars forming $H$ can be built without using any of the vertices from the paths.
We add those paths to the stars found in the solution for $\mathcal{I}'$ and terminate the algorithm with a positive answer.
If no instance of \textsc{\pddpshort} is solved with a positive answer, we conclude that $\mathcal{I}$ is a negative instance.

The running time is $n^{\Ocal(k + r + w)}$, where $w = \dtw(D)$, and is given by the number of choices of $\xarray$, \autoref{corollary:rooted-stars-subdigraph-and-slacks}, and \autoref{theorem:algo-for-the-paths}.
\begin{theorem}\label{thm:our-XP-algo}
The {\sc \pname} problem is solvable in time $n^{\Ocal(k + r + w)}$, where $w$ is the directed treewidth of the input digraph $D$, $k$ is the number of stars forming $H$, and $r$ is the number of paths in the collection $\mathcal{P}$. 
\end{theorem}

\section{Hardness results}\label{sec:hardness_results}
In this section we prove our hardness results for a number of particular cases of {\sc Subdigraph Isomorphism}.
Formally, this problem is defined as follows.

\begin{boxproblem}[framed]{Subdigraph Isomorphism}
\label{problem:subdigraph_isomorphism}
Input: & A host digraph $D$ and a target digraph $H$.\\
Output: & A subdigraph $H' \subseteq D$ such that $H' \cong H$ if it exists, or a report that it does not, otherwise.
\end{boxproblem}
In general, {\sc Subdigraph Isomorphism} is \NP-complete even if the host digraph $D$ is a DAG, from a simple reduction from \textsc{Clique}.
Indeed, given an undirected graph $G$, one can choose an arbitrary ordering of $V(G)$ and build a DAG $D$ by orienting every edge of $G$ as a forward arc with relation to the ordering.
Then a set $X$ is a clique in $G$ if and only if $X$ is a transitive tournament in $D$.
We show that {\sc Subdigraph Isomorphism} remain \NP-complete even if $H$ belongs to very restricted classes of digraphs of bounded breakability and/or maximum degree.
Our hardness results are enumerated in the following theorem.

\begin{theorem}
\label{theorem:hardness results}
Let $k \geq 1$ be an integer.
{\sc Subdigraph Isomorphism} is \NP-complete if
\begin{enumerate}
  \item the host digraph $D$ is a DAG and $H$ is an antidirected path;
  \item the host digraph $D$ has $\dtw(D) = 1$ and $H$ is the expansion of a transitive tournament (and thus of maximum degree at most six);
  \item the host digraph $D$ is a DAG and $H$ is formed by the disjoint union of $2$-out-stars or the disjoint union of $2$-in-stars;
  \item the host digraph $D$ is a DAG and $H$ is formed by the disjoint union of $k$ copies of $2$-in-stars ($2$-out-stars) plus a $k$-homogeneous star whose leaves are all the $k$ centers of the other stars;
  \item the host digraph $D$ is a DAG and $H$ is formed by the disjoint union of digraphs $H_1$ and $H_2$ where each $H_i$ is built from a  $k$-homogeneous star where each arc is $i$-subdivided;
  \item the host digraph $D$ is a DAG and $H$ is caterpillar ending in a $k$-homogeneous star and every other branching vertex has exactly two leaves as in-neighbors or out-neighbors.
\end{enumerate}
\end{theorem}

\noindent In \autoref{table:summary-of-results} we show the implications of our hardness results to parameterized versions of {\sc Subdigraph Isomorphism}.
See \autoref{fig:hardness-digraphs-examples} for examples of digraphs in the classes mentioned in \autoref{theorem:hardness results}.

\begin{table}[h]
\begin{tabularx}{.99\textwidth}{l | >{\centering\arraybackslash}c >{\centering\arraybackslash}c >{\centering\arraybackslash}c}
\toprule
  $H$ & $\breakability_w(H)$ & $\Delta(H)$ & Hardness \\ 
  \midrule
  Transitive tournaments & $1$ & $|V(H)-1|$ & \NP-complete in DAGs\\
  Item \lipItem{1.} of \autoref{theorem:hardness results} & $\infty$ & $2$ & \NP-complete in DAGs\\
  Item \lipItem{2.} of \autoref{theorem:hardness results} & $4^w$ & $7$ & \NP-complete if $\dtw(D) \leq 3$\\
  Item \lipItem{3.} of \autoref{theorem:hardness results} & $\infty$ & $2$ & \NP-complete in DAGs\\
  Items \lipItem{4.}, \lipItem{5.}, and \lipItem{6.} of \autoref{theorem:hardness results} & $\infty$ & $k$ & \NP-complete in DAGs\\
\bottomrule 
\end{tabularx}
\caption{Implications of hardness results for {\sc Subdigraph Isomorphism} parameterized by the breakability and/or the maximum degree of each choice of the target digraph mentioned in \autoref{theorem:hardness results}. In each case, $k$ is as in the statement of the theorem and $w$ is an non-negative integer.}
\label{table:summary-of-results}
\end{table}

The $w$-breakability of transitive tournaments is one because the underlying graph of a tournament is complete.
For antidirected paths, the breakability is unbounded since only adjacent vertices are within a directed path of an antidirected path.
Thus it suffices to take a $0$-guarded set $X$ intersecting only every other vertex of an antidirected path $P$ to see that $V(P) \cap X$ can have as many weak components as $|X|$.
For the breakability of $H$ as defined in items \lipItem{3.} to \lipItem{6.} it suffices to see that, in each of those cases, there are no paths between distinct leaves of $H$.

For item \lipItem{2.}, the proof that the breakability of expansions of transitive tournaments is bounded is not trivial and shown in \autoref{lem:breakability-of-expanded-TTk}.
In the remaining of this section we discuss and prove all the items of \autoref{theorem:hardness results}.

\medskip
\noindent\textbf{Antidirected paths.}
Bang-Jensen et al.~\cite{BANGJENSEN201768} showed that deciding if a given digraph $D$ contains an antidirected path between a given pair of vertices $s,t \in V(D)$ is \NP-complete.
This immediately implies that {\sc Subdigraph Isomorphism} is \NP-complete when $H$ is an antidirected path since one can test each of the $\binom{n}{2}$ instances, one for each pair of vertices of $D$.
In other words, in this particular case the rooted version of the problem generalizes the unrooted one.
Their reduction is not for DAGs, however, and next we show how it can be changed to prove item \lipItem{1.} of \autoref{theorem:hardness results}.
As in their case, we use the following auxiliary result.
\begin{proposition}[Bang-Jensen et al.~{\cite[Theorem 2.1]{BANGJENSEN201768}}]
Let $G$ be an undirected graph and $u,v \in V(G)$.
Given pairs of distinct edges $\mathcal{S} = \{\{e_1, f_1\}, \ldots, \{e_p, f_p\}\}$ of $G$ with $p \geq 1$, it is \NP-complete to decide if $G$ has a path from $s$ to $t$ avoiding at least one edge from each pair in $\mathcal{S}$.
\end{proposition}

If $u,v$ are vertices of a digraph $D$, by \emph{identifying} of $u$ and $v$ we build from $D$ a new digraph $D'$ with vertex set $(V(D) \setminus \{u,v\}) \cup \{x\}$ where the new added vertex $x$ absorbs every arc from $D$ incident to $u$ or $v$, respecting the original orientations.
For convenience, we delete the possibly generated loop on vertex $x$.

\begin{lemma}\label{lemma:antidirected-path-dags}
Let $D$ be a DAG and $s,t \in V(D)$.
It is \NP-complete to decide if there is an antidirected path between $s$ and $t$ in $D$.
\end{lemma}
\begin{proof}
Let $G$ be an undirected graph and $\mathcal{S} = \{\{e_1, f_1\}, \ldots, \{e_p, f_p\}\}$ be a set of pairs of distinct edges of $G$ with $p \geq 1$.
Let $k$ be the maximum number of pairs in $\mathcal{S}$ sharing the same edge and $P$ be the antidirected path of length $6k+2$ starting on source vertex.
Thus $P$ has $6k+3$ vertices and ends in a source vertex.
For every $e \in E(G)$ with extremities $u,v$ add to $D$ a private copy $P_e$ of $P$ with first vertex $u$ and last vertex $v$.
We now distinguish three subsets of $V(D)$.
The \emph{original} vertices are the vertices that also exist in $G$.
The \emph{internal} vertices are $V(D) \setminus V(G)$.
The \emph{usable} vertices are determined as follows.
First, we enumerate the internal vertices of $P_e$ as $a_1, \ldots, a_{6k+1}$ as they appear in the antidirected path from $u$ to $v$.
By construction, every $a_j$ is a sink if $j$ is odd and a source if $j$ is even.
For the usable vertices of $P_e$, we choose the set $\{a_{3j+1} \mid j \in \{0, \ldots, 2k\}\}$.
In other words, we choose $a_1$, which is a sink, skip the following two, then take the source $a_4$, and continue like this until we choose $a_{6k+1}$, which is a sink.
Thus $P_e$ has $2k+1$ usable vertices, $k+1$ sinks and $k$ sources, and the crucial property to ensure that the construction yields a DAG in the end is that all usable vertices of $P_e$ are at more than two indexes of distance from one another.

We are now ready to finish the construction. 
For each pair $\{e_i, f_i\} \in \mathcal{S}$ with $i \in [p]$, we identify one usable sink $P_{e_i}$ with one usable source of $P_{f_j}$.
Thus the vertex resulting from this identification has in-degree and out-degree two.
The choice of the length of $P$ implies that we can identify in pairs without repetition.

The proof now follows exactly as in the original proof of~\cite[Theorem 2.2]{BANGJENSEN201768}.
We claim that $D$ has an antidirected path from $s$ to $t$ if and only if $G$ has a path from $s$ to $t$ which uses at most one edge from each of the pairs in $\mathcal{S}$.
Indeed, let $Q$ be a path from $s$ to $t$ in $G$ using edges $h_1, h_2, \ldots, h_r$ in this order.
Then the concatenation of antidirected paths $P_{h_1}, P_{h_2}, \ldots, P_{h_r}$ is an antidirected path (notice that every original vertex is a source in $D$), since the identifications done in the construction of $D$ were done only for paths corresponding to pairs in $\mathcal{S}$.
For the other direction, let $Q'$ be an antidirected path from $s$ to $t$ in $D$.
By construction, all internal vertices have in-degree and out-degree at most two and the only vertices with $\deg^+_D(v) = \deg^-_D(v) = 2$ are the vertices resulting from the identifications done in the construction of $D$.
Thus, if $u$ is such a vertex, built from an identification associated with the pair $\{e_i, f_i\} \in \mathcal{S}$, then $u$ has both its in-neighbors in $P_{e_i}$ and both its out-neighbors in $P_{f_i}$.
This implies that for every edge $e$ appearing in a pair in $\mathcal{P}$, the antidirected path $Q'$ either traverses $P_{e}$ entirely or avoids all the internal vertices of $P_{e}$.
Now, observing the edges of $G$ associated with the antidirected paths forming $Q'$, we can construct a path in $G$ from $s$ to $t$.
It remains to prove that $D$ is a DAG.

By contradiction, assume that $C$ is a cycle of $D$.
Clearly $C$ cannot contain an original vertex since all such vertices are sources.
In fact, $C$ has to be formed entirely by vertices created by the identification of usable vertices in the construction of $D$, for only those are not sinks nor sources.
This implies that $C$ contains a directed path between two usable vertices in the same antidirected path $P_e$, contradicting the choice of usable vertices.
Thus $D$ is free of cycles and the result follows.
\end{proof}

The difference between our proof and the one of~\cite[Theorem 2.2]{BANGJENSEN201768} is in the choice of the vertices which can be used in the identifications.
By taking a large antidirected path $P$ (with length $6k+2$ instead of $2k+2$), we ensure that there is no path between usable vertices of $P$ and thus no cycle can appear in $D$.
In \autoref{fig:cycles-in-antidirected-construction} we show an example of a cycle that can appear in the original construction from~\cite{BANGJENSEN201768}.

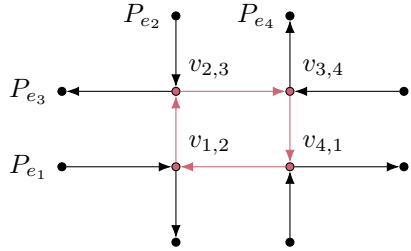
\begin{figure}[h]
\centering
\scalebox{1}{\begin{tikzpicture}[xscale = 1.5]
\node[blackvertex, scale=.6, label=45:{$v_{1,2}$}, fill=goodred] (v12) at (0,0) {};
\node[blackvertex, scale=.6, label=45:{$v_{2,3}$}, fill=goodred] (v23) at (0,1) {};
\node[blackvertex, scale=.6, label=45:{$v_{3,4}$}, fill=goodred] (v34) at (1,1) {};
\node[blackvertex, scale=.6, label=45:{$v_{4,1}$}, fill=goodred] (v41) at (1,0) {};

\node[blackvertex, scale=.6, label = 180:$P_{e_1}$] (pe1) at ($(v12) + (-1,0)$) {};
\node[blackvertex, scale=.6] (pe1x) at ($(v41) + (1,0)$) {};
\draw[arrow] (pe1) to (v12);
\draw[arrow, goodred] (v41) to (v12);
\draw[arrow] (v41) to (pe1x);

\node[blackvertex, scale=.6, label = 180:$P_{e_2}$] (pe2) at ($(v23) + (0,1)$) {};
\node[blackvertex, scale=.6] (pe2x) at ($(v12) + (0,-1)$) {};
\draw[arrow] (pe2) to (v23);
\draw[arrow, goodred] (v12) to (v23);
\draw[arrow] (v12) to (pe2x);

\node[blackvertex, scale=.6, label = 180:$P_{e_3}$] (pe3) at ($(v23) + (-1,0)$) {};
\node[blackvertex, scale=.6] (pe3x) at ($(v34) + (1,0)$) {};
\draw[arrow] (v23) to (pe3);
\draw[arrow, goodred] (v23) to (v34);
\draw[arrow] (pe3x) to (v34);

\node[blackvertex, scale=.6, label = 180:$P_{e_4}$] (pe4) at ($(v34) + (0,1)$) {};
\node[blackvertex, scale=.6] (pe4x) at ($(v41) + (0,-1)$) {};
\draw[arrow] (v34) to (pe4);
\draw[arrow, goodred] (v34) to (v41);
\draw[arrow] (pe4x) to (v41);
\end{tikzpicture}}%
\caption{Example of a cycle resulting from the original construction of \cite[Theorem 2.2]{BANGJENSEN201768}, where the antidirected path $P$ has length $2k+2$. Notice the cycle in the center (in red). In the example, the ordered pairs of edges are $\{\{e_1, e_2\}, \{e_2, e_3\}, \{e_3, e_4\}, \{e_4, e_1\}\}$. Each $v_{i,j}$ with $i,j \in [4]$ in the figure appears as the result of the identification associated with the pair $\{e_i, e_j\}$.}
\label{fig:cycles-in-antidirected-construction}
\end{figure}

\medskip
\noindent\textbf{Expansions of transitive tournaments.}
The hardness proof of this case is as easy as the reduction from \textsc{Clique} to the case when $H$ is a transitive tournament.
Namely, given an instance of \textsc{Clique} with input undirected graph $G$, we order the vertices of $G$ arbitrarily and build a digraph $D$ on $V(G)$ by adding an arc from $u$ to $v$ in $D$ whenever $\{u,v\} \in G$ and $u$ precedes $v$ in the order.
Let $D'$ be the expansion of $D$.
We remind the reader of the loops in the original vertices of $D'$. 
Now, $G$ contains a clique of size $k$ if and only if $D'$ contains the expansion of a transitive tournament of size $k$.
To prove this, it suffices to see that there is an edge $\{u,v\} \in E(G)$ with $u$ preceding $v$ if and only if there is a path from $u$ to $v$ in $D'$, where $u$ and $v$ are original vertices, and the result follows.

Now, we show that the directed treewidth of expansions of DAGs (so, in particular, of transitive tournaments) is at most three.
For this, we use the following result by Kintali~\cite{KINTALI201783}.
\begin{proposition}[Kintali~\cite{KINTALI201783}]
\label{proposition:directed-treewidth-and-circumference}
Let $D$ be a digraph and $k$ be the maximum length of a cycle in $D$.
Then $\dtw (D) \leq k + 1$.
\end{proposition}

We also adopt the following notations.
If $D'$ is an expansion of a digraph $D$ and $v$ is an internal vertex of $D'$, we denote by $\sigma(v)$ the original vertex of $D'$ that, when processed in the construction of $D'$, resulted in the addition of $v$ to $D'$.
Thus, if $v$ is a type-1 internal vertex then there is a path from $\sigma(v)$ to $v$, and if $v$ is a type-2 internal vertex then there is a path from $v$ to $\sigma(v)$.
If $v$ is an original vertex, we define $\sigma(v) = v$.

Additionally, for an original vertex $u \in V(D')$, we denote by $V^+_u$ the vertices of the out-arborescence with root $u$ added to $D'$ when the out-neighborhood of $u$ in $D'$ is processed in the first part of the construction.
Thus $V^+_u$ has only type-1 internal vertices.
Similarly, we denote by $V^-_u$ the vertices of the in-arborescence with root $u$ added to $D'$ when the in-neighborhood of $u$ in $D'$ is processed in the second part of the construction.
Thus $V^-_u$ has only type-2 internal vertices.

\begin{lemma}\label{lem:directed-treewidth-expansion-of-dags}
Let $D$ be a DAG and $D'$ be the expansion of $D$.
Then the directed treewidth of $D'$ is at most three.
\end{lemma}
\begin{proof}
We show that, in $D'$, there are no directed cycles of length greater than two.
Let $u$ be an internal vertex.
By contradiction, assume that $C$ is a cycle of length at least three in $D'$. 
Then $C$ contains both type-1 and type-2 internal vertices.
This holds due to the fact that $C$ cannot be contained entirely in $V^+_u$ nor $V^-_u$ for any original vertex $u$.

Let $u$ be an original vertex such that $V^+_u$ contains a vertex $w \in V(C)$.
In this case, by the discussion above, when following $C$ starting from $w$, at some point the path leaves $V^+_u$ using a path to a type-2 vertex $w'$.
Then in $D$ there is a path from $u$ to the original vertex $\sigma(w')$.
Now, the path in $C$ from $w'$ back to $w$ has to reach $u$ in order to get back into $V^+(u)$.
Since there are no arcs between type-2 and type-1 internal vertices, this implies that $C$ contains a path from $w'$ to an original vertex $v$ (it is possible that $v = \sigma(w')$) and then back to $w$.
Thus $D$ has a path from $u$ to $\sigma(w')$, then to $v$ and back to $u$, contradicting the fact that $D$ is a DAG.

We conclude that there are no cycles of length greater than two in $D'$ and the result follows by applying \autoref{proposition:directed-treewidth-and-circumference}.
\end{proof}

It remains to bound the breakability of expansions of transitive tournaments, which is the most interesting part of our reduction. 

\begin{lemma}\label{lem:breakability-of-expanded-TTk}
Let $T$ be a transitive tournament, $T'$ be the expansion of $T$, and $w$ be a non-negative integer.
Then $\breakability_w(T') \leq 4^w$.
\end{lemma}

\begin{proof}
We begin with the following claim.
\begin{claim}\label{claim:expansions-of-tournaments-almost-unilateral}
For every distinct $u,v \in V(T')$, there is a path in $T'$ from $u$ to $v$ or from $v$ to $u$ (or both).
\end{claim}
\begin{proof}[Proof of the claim.]
We consider three cases.
If both $u$ and $v$ are original vertices, the claim follows since $T$ is a transitive tournament.

Assume now that both $u$ and $v$ are internal vertices.
If $u$ is a type-1 internal vertex and $\sigma(u) = \sigma(v)$ the claim follows since the subdigraph of $T'$ induced by $V^+_{\sigma(u)} \setminus \{\sigma(u)\}$ is strongly connected. 
The argument is similar if $u$ is a type-2 internal vertex.
Assume now that $\sigma(u) \neq \sigma(v)$.

Assume that in $T$ there is an arc from $\sigma(u)$ to $\sigma(v)$.
If $u$ is a type-1  internal vertex, then the construction of $D'$ ensures that there is a path from $u$ to a vertex in $V^-_{\sigma(v)}$ and finally to $v$.
Notice that this path uses $\sigma(v)$ only if $v$ is a type-2 internal vertex.
If $u$ is a type-2 internal vertex, then there is a path from $u$ to $\sigma(u)$, then to $V^-_{\sigma(v)}$.
Similarly to the previous case, this implies the existence of a path from $u$ to $v$ using $\sigma(v)$ only if $v$ is a type-2 internal vertex.
If the arc in $T$ is from $\sigma(v)$ to $\sigma(u)$ then we find a path in $T'$ from $v$ to $u$ with an analogous proof.
Finally, assume that $u$ is an internal vertex and $v$ is an original vertex.
Then there is an arc in $T$ with endpoints $u$ and $\sigma(v)$.
The analysis is similar to the case when both $u$ and $v$ are internal vertices, and the claim follows.
\end{proof}

\autoref{claim:expansions-of-tournaments-almost-unilateral} implies that, if $X \subseteq V(T')$, then there is a path in $T'$ between any two weak components of the subdigraph induced by $X$.
Thus if $X$ is $Z$-guarded by some $Z \subseteq V(T')$, then every such path between weak components of $T'[X]$ must be intersected by $Z$.
The goal now is to show that, from every $v \in V(T')$, we can partition $V(T') \setminus \{v\}$ in at most four sets $\{X^v_i \mid 1 \leq i \leq 4\}$ such that for  any $1 \leq i \leq 4$ and any $x,y \in X^v_i$, there is a path from $x$ to $y$ not containing~$v$, or a path from $y$ to $x$ not containing~$v$.

Let $v \in V(T')$.
We describe how to partition $V(T') \setminus \{v\}$ when $v$ is a type-1 internal vertex and when $v$ is an original vertex.
The case when $v$ is a type-2 internal vertex is analogous to the type-1 case.

Let $v$ be a type-1 internal vertex.
First, let $u_1$ be the in-neighbor of $v$ in the shortest path $P$ from $\sigma(v)$ to $v$ in $T'$.
Notice that it is possible that $u_1 = \sigma(v)$.
If $\deg^+_{T'}(\sigma(v)) = 2$ we denote by $z$ the out-neighbor of $\sigma(v)$ in $T'$ that is not in $P$.
By construction, $v$ has at most two neighbors other than $z$ (if $u_1 = \sigma(v)$ and $z$ is defined) and $u_1$.
Let $u_2$ be one of those neighbors, which is guaranteed to exist by construction, and $u_3$ be the second one if it exists.

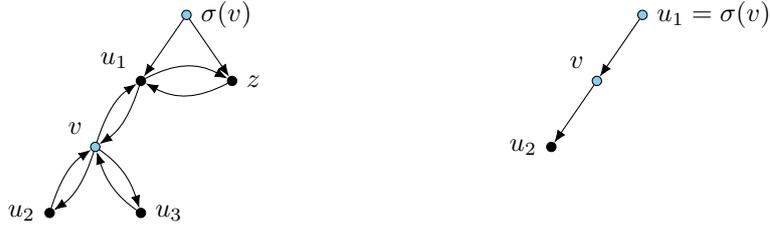
\begin{figure}[h]
\centering
\scalebox{1}{\begin{tikzpicture}[xscale=1.2, yscale=1.75]
\begin{scope}
\node[bluevertex, scale=.7, label=0:$\sigma(v)$] (orig) at (0,0) {};
\node[blackvertex, scale=.7, label=135:$u_1$] (u1) at ($(orig) + (-.5, -.5)$) {};
\node[blackvertex, scale=.7, label=0:$z$] (z) at ($(orig) + (.5, -.5)$) {};

%\node[blackvertex, scale=.7, label=180:$u_1$] (u1) at ($(orig) + (-.5, -.5)$) {};
\node[bluevertex, scale=.7, label=135:$v$] (v) at ($(u1) + (-.5, -.5)$) {};

\node[blackvertex, scale=.7, label=180:$u_2$] (u2) at ($(v) + (-.5, -.5)$) {};
\node[blackvertex, scale=.7, label=0:$u_3$] (u3) at ($(v) + (.5, -.5)$) {};

\foreach \u/\v in {u1/z, u1/v, v/u2, v/u3} {
    \draw[arrow] (\u) to [bend left = 20] (\v);
    \draw[arrow] (\v) to [bend left = 20] (\u);
}
\draw[arrow] (orig) to (u1);
\draw[arrow] (orig) to (z);
\end{scope}

\begin{scope}[xshift=5cm]
\node[bluevertex, scale=.7, label=0:{$u_1 = \sigma(v)$}] (orig) at (0,0) {};
\node[bluevertex, scale=.7, label=135:$v$] (v) at ($(orig) + (-.5, -.5)$) {};
\node[blackvertex, scale=.7, label=180:$u_2$] (u2) at ($(v) + (-.5, -.5)$) {};
%\node[blackvertex, scale=.7, label=0:$u_3$] (u3) at ($(v) + (.5, -.5)$) {};

\foreach \u/\v in {v/u2} {
    \draw[arrow] (\u) to  (\v);
    %\draw[arrow] (\v) to [bend left = 20] (\u);
}

\draw[arrow] (orig) to (v);
%\node[blackvertex, scale=.7, label=0:$z$] (z) at ($(orig) + (.5, -.5)$) {};

%\node[blackvertex, scale=.7, label=180:$u_1$] (u1) at ($(orig) + (-.5, -.5)$) {};
%\node[bluevertex, scale=.7, label=135:$v$] (v) at ($(u1) + (-.5, -.5)$) {};

%\foreach \u/\v in {u1/z, u1/v, v/u2, v/u3} {
%    \draw[arrow] (\u) to [bend left = 20] (\v);
%    \draw[arrow] (\v) to [bend left = 20] (\u);
%}
%\draw[arrow] (orig) to (u1);
%\draw[arrow] (orig) to (z);
\end{scope}

\end{tikzpicture}}%
\caption{On the left, the selection of vertices $u_1, u_2, u_3$ and, $z$ when all of them exist. In the proof, it is possible that $u_1 = \sigma(v)$ and that $u_3$ and/or $z$ are non-existent. The case when both are non-existent is depicted on the right. In this case, $u_2$ is also an original vertex. $v$ and $\sigma(v)$ are in blue in both figures.}
\label{fig:partitioning-expansion-of-tournament-breakability-proof}
\end{figure}

We are ready to partition $V(T')$.

\begin{itemize}
  \item We denote by $X^v_2$ the set of all vertices $x \in V(T')$ such that there is a path in $T'$ from $u_2$ to $x$ avoiding $v$.
  \item If $u_3$ is defined, we denote by $X^v_3$ the set of all vertices $x \in V(T') \setminus X^v_2$ such that there is a path in $T'$ from $u_3$ to $x$ avoiding $v$. Otherwise, we set $X^v_3 = \emptyset$.
  \item We denote by $X^v_4$ the set of all vertices $x \in V(T')$ such that there is a path from $x$ to $\sigma(v)$.
  \item If $z$ is defined, we denote by $X^v_1$ the set of all vertices $x \in V(T') \setminus (X^v_2 \cup X^v_3)$ such that there is a path from $z$ to $x$ in $T'$ avoiding $v$.
  Otherwise, we set $X^v_1 = \emptyset$.
\end{itemize}

The vertex $u_1$ is an exception case.
If $u_1 = \sigma(v)$ then it was included in $X^v_4$.
Otherwise, if $z$ is defined it was included in $X^v_1$.
If none of those cases occurs, then we include $u_1$ in $X^v_1$.
In this case, $|X^v_1| = 1$.
Notice that, since $T$ is a DAG, no in-neighbors of $\sigma(v)$ in $T$ are in $X^v_2 \cup X^v_3 \cup X^v_4$.
Moreover, the construction of $X^v_2, X^v_3$, and $X^v_4$ ensures that the sets are pairwise disjoint.
Since $T$ is a tournament, every vertex of $T'$ appears in one of those sets and thus $\{X^v_i \mid i \in [4]\}$ defines a partition of $V(T') \setminus \{v\}$ (with some of the sets being possibly empty).

If $v$ is an original vertex, we simply split $V(T') \setminus \{v\}$ into two sets $X^v_1, X^v_2$.
The first contains every vertex $x \in V(T')$ other than $v$ such that there is a path from $x$ to $v$ in $T'$, and $X^v_2 = (V(T') \setminus X^v_1) \setminus \{v\}$.
In this case, we set $X^v_3 = X^v_4 = \emptyset$.

In any case, by construction, for every $y \neq v$, it holds that $V^+_{\sigma(y)} \cup V^-_{\sigma(y)}$ is entirely contained in one set $X^v_i$ with $i \in [4]$.
We now show that for each $i \in [4]$, between vertices $x,y$ contained in a set $S^v_i$, one can always find a path from $x$ to $y$ or from $y$ to $x$ (or both) avoiding $v$.
\begin{claim}\label{claim:paths-avoiding-v-in-partition-of-expansions}
Let $C_1, C_2$ be non-empty and disjoint subsets of $V(T')$ with $C_1, C_2 \subseteq X^v_i$, for some $i \in [4]$.
Then there is a path from $C_1$ to $C_2$ or from $C_2$ to $C_1$ not using $v$.
\end{claim}
\begin{proof}[Proof of the claim.]
If  $C_1$ intersects some $V^+_s \cup V^-_s$ and $C_2$ intersects $V^+_t \cup V^-_t$ for some original vertices $s,t \neq \sigma(v)$, then there is a path from  $C_1$ to $C_2$ or from $C_2$ to $C_1$ avoiding $v$, by the construction of $T'$ and since $T$ is a tournament, which guarantees the existence of an arc with endpoints $s$ and $t$ in $T$ if $s \neq t$.
The same holds if $C_1$ and $C_2$ intersect $V^+_{\sigma(v)}$ since $V^+_{\sigma(v)} \cap X^v_i$ induces a strongly connected digraph, and the same holds for $V^-_{\sigma(v)}$.
The only remaining cases are when one of $C_1, C_2$ intersects $V^+_{\sigma(v)}$ or $V^-_{\sigma(v)}$ (notice that $V^+_{\sigma(v)}$ and $V^-_{\sigma(v)}$ cannot intersect the same set $X^v_i$ since $T$ is a DAG), and the other intersects $V^+_{s} \cup V^-_{s}$, for some $s \neq v$.
The analysis of those cases is similar, and the claim follows.
\end{proof}

Now let $X$ be a $w$-guarded set of $T'$.
By induction on $k$, we prove that the number of weak components of $T'[X]$ is at most $4^k$ when $X$ is $k$-guarded.

For the base case, assume that $k = 1$ and that $X$ is guarded by a set $Z = \{v\}$.
Let $\{X^v_i \mid i \in [4]\}$ define a partition of $V(T') \setminus \{v\}$, as previously defined.
If $D[X \cap V(T')]$ has  at least five weak components, then at least two of them, say $C$ and $C'$, are contained in the same set $X^v_j$ for some $j \in [4]$. 
Then,  by~\autoref{claim:paths-avoiding-v-in-partition-of-expansions}, there is a path in $T'$ between a vertex in $V(C)$ and a vertex in $V(C')$ avoiding $v$, contradicting our assumption that $X$ is $Z$-guarded.
Thus $D[X \cap V(T')]$ has at most four weak components and the base case holds.

For the induction step, assume that the result holds for $k \geq 1$, that $X$ is guarded by a vertex-minimal set $Z$ with $|Z| = k+1$, and let $v \in Z$.
As in the previous case, $\{X^v_i \mid i \in [4]\}$ defines a partition of $V(T') \setminus \{v\}$.
For $i \in [4]$ let $c_i$ be the number of weak components of $D[X \cap V(T') \cap X^v_i]$.
Since $X \cap X^v_i$ is $k$-guarded (as $v \not \in X^v_i)$, the induction hypothesis implies that $c_i \leq 4^{k}$.
Thus the number of weak components of $T'[X]$ is at most $\sum_{i \in [4]}c_i \leq 4^{k+1}$, and the lemma follows.
\end{proof}

\medskip
\noindent\textbf{Disjoint union of $2$-out-stars or $2$-in-stars.}
For this case, we present a reduction from an \NP-complete matching problem in bipartite graphs.
\begin{proposition}[Plaisted and Zaks~\cite{PLAISTED198065}]\label{proposition:np-complete-matching-bipartite}
Let $G$ be a bipartite graph with vertex partition $\{V_1, V_2\}$ and maximum degree three.
Let $\mathcal{P}_1, \mathcal{P}_2$ be partitions of $V_1$ and $V_2$, respectively, into parts of size at most two.
It is \NP-complete to decide if $G$ contains a perfect matching $M$ such that no two distinct edges $\{u,v\}, \{u',v'\} \in M$ have the property that $u,u'$ are in the same part of $\mathcal{P}_1$ and $v,v'$ are in the same part of $\mathcal{P}_2$.
\end{proposition}
Adopting terminology from~\cite{PLAISTED198065}, we say that such a matching $M$ is \emph{consistent} with $\mathcal{P}_1$ and $\mathcal{P}_2$.
With this definition, we go straight into the proof.
The key property of \autoref{proposition:np-complete-matching-bipartite} is the bound on the size of the sets forming $\mathcal{P}_1$ and $\mathcal{P}_2$, which allows us to define a finite set of configurations on which we can add the stars that we wish to find in the instance of {\sc Subdigraph Isomorphism}.

\begin{proof}[Proof of item \lipItem{3.} of \autoref{theorem:hardness results}.]
We first prove the case when $H$ is the disjoint union of $2$-out-stars.
Let $G$ be a bipartite graph with vertex partition $\{V_1, V_2\}$ and maximum degree three.
Let $\mathcal{P}_1, \mathcal{P}_2$ be partitions of $V_1$ and $V_2$, respectively, into parts of size at most two.
To build the DAG $D$, we first orient ever edge of $G$ from $V_1$ to $V_2$.
Then, for each $e \in E(D)$ with $e = \{u,v\}$, add to $D$ a vertex $c(e)$ with out-neighbors $u$ and $v$.
We can assume that $|V_1| = |V_2|$ since otherwise there is no perfect matching in $G$.
The idea is that by taking stars with centers $c_e$ we are selecting non-adjacent edges of $G$.
However, at this point nothing forbids a hypothetical algorithm from choosing two edges breaking the consistency property that we want to maintain.
The strategy then is to add more stars.
Namely, we start with a $k' = 0$ and in each of the four cases of the construction, distinguished by how many arcs there are between parts in $\mathcal{P}_1$ and $\mathcal{P}_2$, we add new stars in a way that ensures the identification of a matching that is consistent with the partitions.
In each case, we update the value of $k'$ accordingly and at the end ask for $|V_1| + k'$ stars.
The extra $k'$ \emph{does not} guarantee that a consistent and perfect matching is found but we show that, even if that is not the case, we can still efficiently extract such a matching from the outputted stars.
In order to keep track of the two types of stars, the ones corresponding to edges of $G$ and the ones ensuring that good choices are possible, we introduce a set $F$ of vertices which can increase in size depending on the cases below.
We start with $F = \emptyset$.

Now for each $A,B$ with $A \in \mathcal{P}_1$ and $B \in \mathcal{P}_2$, we denote by $E(A, B)$ the set of arcs from $A$ to $B$ in $D$ and consider four cases in the construction.
See~\autoref{fig:construction-union-of-2-stars} for an illustration of the four cases.

\begin{enumerate}
  \item $|E(A, B)| \leq 1$. In this case, we maintain $k'$ as it is, and $F$ as it is.
  \item $|E(A, B)| = 2$. Let $e_1, e_2$ be the arcs from $A$ to $B$ in $D$. We add to $F$ a new vertex $v$ with three out-neighbors: $c(e_1), c(e_2)$, and a newly added vertex $v'$. We increase $k'$ by one.
  \item $|E(A, B)| = 3$. Let $e_1, e_2$, and $e_3$ be the arcs from $A$ to $B$ in $D$. We add to $F$ a new vertex $v$ with three out-neighbors $c(e_1), c(e_2)$, and $c(e_3)$. We increase $k'$ by one.
  \item $|E(A, B)| = 4$. Let $\{e_i \mid i \in [4]\}$ be the arcs from $A$ to $B$ in $D$. We add to $F$ two new vertices $v_1, v_2$ both with five out-neighbors: the set $\{c(e_i) \mid i \in [4]\}$ plus a newly added vertex $v'$. We increase $k'$ by two.
\end{enumerate}

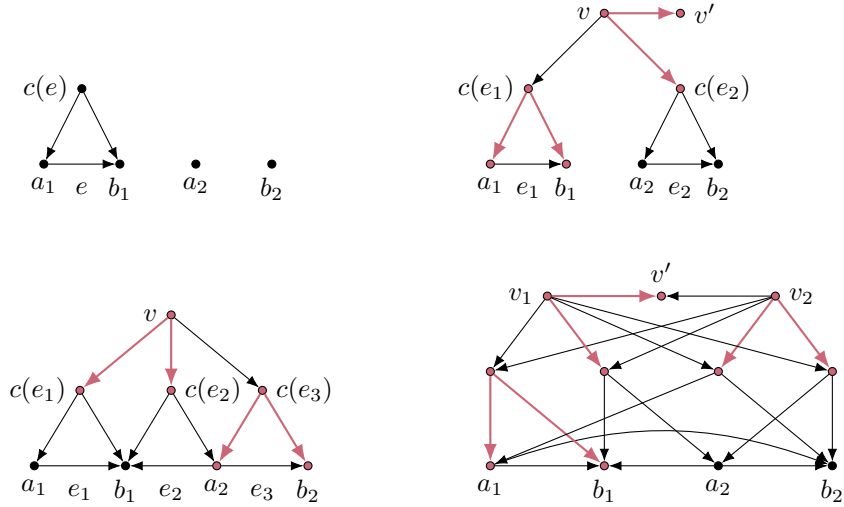
\begin{figure}[h]
\centering
\scalebox{1}{\begin{tikzpicture}
\begin{scope}[xshift=.125cm]
\node[blackvertex, scale=.6, label=-90:$a_1$] (a1) at (0,0) {};
\node[blackvertex, scale=.6, label=-90:$b_1$] (b1) at ($(a1) + (1,0)$) {};
\node[blackvertex, scale=.6, label=-90:$a_2$] (a2) at ($(b1) + (1,0)$) {};
\node[blackvertex, scale=.6, label=-90:$b_2$] (b2) at ($(a2) + (1,0)$) {};
\draw[arrow] (a1) to node[midway, label=-90:$e$] {} (b1);
\node[blackvertex, scale=.6, label=180:{$c(e)$}] (c) at ($(a1) + (0.5, 1)$) {};
\draw[arrow] (c) to (a1);
\draw[arrow] (c) to (b1);
\end{scope}%
\begin{scope}[xshift=6cm]
\node[redvertex, scale=.6, label=-90:$a_1$] (a1) at (0,0) {};
\node[redvertex, scale=.6, label=-90:$b_1$] (b1) at ($(a1) + (1,0)$) {};
\node[blackvertex, scale=.6, label=-90:$a_2$] (a2) at ($(b1) + (1,0)$) {};
\node[blackvertex, scale=.6, label=-90:$b_2$] (b2) at ($(a2) + (1,0)$) {};
\draw[arrow] (a1) to node[midway, label=-90:$e_1$] {} (b1);
\draw[arrow] (a2) to node[midway, label=-90:$e_2$] {} (b2);
\node[redvertex, scale=.6, label=180:{$c(e_1)$}] (c1) at ($(a1) + (0.5, 1)$) {};
\node[redvertex, scale=.6, label=0:{$c(e_2)$}] (c2) at ($(a2) + (0.5, 1)$) {};
\node[redvertex, scale=.6, label=180:$v$] (v) at ($(c1) + (1, 1)$) {};
\node[redvertex, scale=.6, label=0:$v'$] (vprime) at ($(v) + (1, 0)$) {};
\foreach \i/\j in {2/2} {
	\draw[arrow] (c\i) to (a\j);
	\draw[arrow] (c\i) to (b\j);
}
\foreach \i/\j in {1/1} {
	\draw[arrow, thick, goodred] (c\i) to (a\j);
	\draw[arrow, thick, goodred] (c\i) to (b\j);
}
\foreach \i in {1}{
	\draw[arrow] (v) to (c\i);
}
\draw[arrow, thick, goodred] (v) to (c2);
\draw[arrow, thick, goodred] (v) to (vprime);
\end{scope}%
\begin{scope}[yshift=-4cm, xscale=1.2]
\node[blackvertex, scale=.6, label=-90:$a_1$] (a1) at (0,0) {};
\node[blackvertex, scale=.6, label=-90:$b_1$] (b1) at ($(a1) + (1,0)$) {};
\node[redvertex, scale=.6, label=-90:$a_2$] (a2) at ($(b1) + (1,0)$) {};
\node[redvertex, scale=.6, label=-90:$b_2$] (b2) at ($(a2) + (1,0)$) {};
\draw[arrow] (a1) to node[midway, label=-90:$e_1$] {} (b1);
\draw[arrow] (a2) to node[midway, label=-90:$e_2$] {} (b1);
\draw[arrow] (a2) to node[midway, label=-90:$e_3$] {} (b2);
\node[redvertex, scale=.6, label=180:{$c(e_1)$}] (c1) at ($(a1) + (0.5, 1)$) {};
\node[redvertex, scale=.6, label=0:{$c(e_3)$}] (c2) at ($(a2) + (0.5, 1)$) {};
\node[redvertex, scale=.6, label=0:{$c(e_2)$}] (c3) at ($(c1) + (1, 0)$) {};
\node[redvertex, scale=.6, label=180:$v$] (v) at ($(c1) + (1, 1)$) {};
\foreach \i/\j in {1/1} {
	\draw[arrow] (c\i) to (a\j);
	\draw[arrow] (c\i) to (b\j);
}
\foreach \i/\j in {2/2} {
	\draw[arrow, thick, goodred] (c\i) to (a\j);
	\draw[arrow, thick, goodred] (c\i) to (b\j);
}
\draw[arrow] (c3) to (b1);
\draw[arrow] (c3) to (a2);
\foreach \i in {1,3}{
	\draw[arrow, thick, goodred] (v) to (c\i);
}
\draw[arrow] (v) to (c2);
\end{scope}%
\begin{scope}[yshift=-4cm, xshift=6cm]
\node[redvertex, scale=.6, label=-90:$a_1$] (a1) at (0,0) {};
\node[redvertex, scale=.6, label=-90:$b_1$] (b1) at ($(a1) + (1.5,0)$) {};
\node[blackvertex, scale=.6, label=-90:$a_2$] (a2) at ($(b1) + (1.5,0)$) {};
\node[blackvertex, scale=.6, label=-90:$b_2$] (b2) at ($(a2) + (1.5,0)$) {};
\draw[arrow] (a1) to (b1);
\draw[arrow] (a2) to (b1);
\draw[arrow] (a2) to (b2);
\draw[arrow] (a1) to [bend left=20] (b2);

\node[redvertex, scale=.6] (c1) at ($(a1) + (0, 1.25)$) {};
\node[redvertex, scale=.6] (c2) at ($(b1) + (0, 1.25)$) {};
\node[redvertex, scale=.6] (c3) at ($(a2) + (0, 1.25)$) {};
\node[redvertex, scale=.6] (c4) at ($(b2) + (0, 1.25)$) {};

\node[redvertex, scale=.6, label=180:$v_1$] (v1) at ($(c1) + (0.75, 1)$) {};
\node[redvertex, scale=.6, label=0:$v_2$] (v2) at ($(v1) + (3, 0)$) {};
\draw[arrow, thick, goodred] (c1) to (a1);
\draw[arrow, thick, goodred] (c1) to (b1);

\draw[arrow] (c2) to (a2);
\draw[arrow] (c2) to (b1);

\draw[arrow] (c3) to (a1);
\draw[arrow] (c3) to (b2);

\draw[arrow] (c4) to (a2);
\draw[arrow] (c4) to (b2);

\foreach \i in {1,2}{
	\draw[arrow] (v2) to (c\i);
}
\foreach \i in {3,4}{
	\draw[arrow, thick, goodred] (v2) to (c\i);
}
\foreach \i in {1,3,4}{
    \draw[arrow] (v1) to (c\i);
}
\draw[arrow, thick, goodred] (v1) to (c2);
\node[redvertex, scale=.6, label=$v'$] (vprime) at ($(v1)+ (1.5,0)$) {};
\draw[arrow, thick, goodred] (v1) to (vprime);
\draw[arrow] (v2) to (vprime);
\end{scope}

\end{tikzpicture}}%
\caption{The four cases of the construction. In each case, $A = \{a_1, a_2\}$ and $B = \{b_1, b_2\}$. In the last case, when $|E(A, B)| = 4$, some of the labels are omitted. The vertices $v, v_1, v_2$ are added to $F$ in each of the relevant cases. Examples of choices for the $2$-out-stars using all vertices of $F$ are shown in red.}
\label{fig:construction-union-of-2-stars}
\end{figure}

Now, we claim the following.
Let $\ell \geq 1$ be an integer. If $D$ has a subdigraph $D'$ formed by the disjoint union of $\ell$ copies of $2$-out-stars, then $D$ has subdigraph $H'$ formed by at least $\ell$ copies of $2$-out-stars such that every vertex in $F$ is the center of one of the out-stars forming $H'$.
To prove the claim, start with $H' = D'$.
To prove the claim, we proceed through each of the four cases described above showing, in each case, how to add $2$-out-stars with centers in $F$, possibly replacing other $2$-out-stars in $H'$ when necessary.
If $H'$ has a $2$-out-star with center $w \in V_1$ using an arc $e$, then we replace this out-star in $H'$ by the $2$-out-star with center $c_e$ and leaves $\{w, z\}$, where $z$ is the head of $e$.
We now assume that none of the $2$-out-stars forming $H'$ has its center in $V_1$.
For each $A, B$ with $A \in \mathcal{P}_1$ and $B \in \mathcal{P}_2$, we proceed as follows, adopting the same terminology for the arcs and vertices in $F$ that we used to describe the cases above.
\begin{enumerate}[(a)]
\item $|E(A, B)| \leq 1$, there is nothing to be done as no vertices of $F$ are associated with this case.
\item $|E(A, B)| = 2$. In the worst case, both $c(e_1)$ and $c(e_2)$ are centers of stars in $H'$, and $v \not \in V(H')$. Then we replace one of those stars, say, the $2$-out-star with center $c(e_2)$, by the $2$-out-star with center $v$ and leaves $v, c(e_2)$.
In particular, if we started with one $c(e_i)$ with $i \in [2]$ not in $V(H')$, then by including in $H'$ a $2$-out-star with center $v$ and leaves $v'$ and $(e_i)$, we increase by one the number of $2$-out-stars forming $H'$ when compared to $D'$.
\item $|E(A, B)| = 3$. In this case, if $v \not \in V(H')$ then at least one $c(e_i)$ with $i \in [3]$ is not used as the center of a $2$-out-star in $H'$.
Then we add to $H'$ the $2$-out-star with center $v$ and leaves $c(e_i), c(e_j)$ with $j \in [2]$ and $j\neq i$, replacing the $2$-out-star with center $c(e_j)$ in $H'$ if it exists.
In particular, if we started with at most one of $c(e_i)$ in $V(H')$, then this procedure increased the number of $2$-out-stars of $H'$ by one.
\item $|E(A, B)| = 4$. If $v_1, v_2 \not \in V(H')$, then at most two $c(e_i), c(e_j)$ with $i,j \in [4]$ are centers of $2$-out-stars in $H'$.
If both these vertices are used as centers in $H'$, by deleting one of them from $H'$ we ensure that three vertices $x,y,z \in \{c(e_i) \mid i \in [4]\}$ are not used in $H'$.
Now we add to $H'$ the $2$-out-stars with centers $v_1, v_2$ using vertices $v',x , y, z$ as the leaves appropriately.
Notice that the number of stars in $H'$ can only increase with relation to $D'$. The case when exactly one of $v_1, v_2$ is already in $V(H')$ follows similarly.
Similarly to the previous case, if we started with at most one $c(e_i)$ in $V(H')$, then this procedure increased the number of $2$-out-stars in $H'$ by two.
\end{enumerate}
In each of the cases, the number of $2$-out-stars in $H'$ can only increase when compared to $D'$ and the claim follows.
These ``good'' configurations are depicted in red in each of the cases of \autoref{fig:construction-union-of-2-stars}.

Let $k = |V_1| + k'$.
We are ready to prove that $G$ has a perfect and consistent matching if and only if $D$ has a subdigraph $H'$ formed by the disjoint union of $k$ copies of $2$-out-stars.
First, assume that $G$ has a perfect and consistent matching $M$ and start with an empty digraph $D'$.
For each $e \in M$, add to $D'$ the $2$-out-star with center $c(e)$ having the endpoints of $e$ as leaves.
At this point, $D'$ has exactly $|V_1|$ pairwise disjoint copies of $2$-out-stars, and no vertex of $F$ is in $D'$.
Applying the aforementioned claim with respect to $|V_1|$ and $D'$, we obtain a subdigraph $H'$ of $D$ formed by $k$ disjoint copies of $2$-out-stars.
The increase from $|V_1|$ to $k$ copies of $2$-out-stars is due to the fact that, in each of the cases \lipItem{(b)}, \lipItem{(c)}, and \lipItem{(d)} above, by the construction of $D'$, at most one vertex of the form $c(e_i)$ with $e_i \in E(A, B)$ is used as the center of a star in $D'$.
Thus each occurrence of the cases \lipItem{(b)}, \lipItem{(c)}, and \lipItem{(d)} increases the number of $2$-out-stars in $H'$ by one, one, and two, respectively, and the necessity follows.

For the sufficiency, assume now that $D$ has a subdigraph $H'$ formed by the disjoint union of $k$ copies of $2$-out-stars.
Again applying our claim, we ensure that every vertex of $F$ is used as the center of some $2$-out-star in $H'$.
Thus the construction of $D$ ensures that, in each of the cases \lipItem{1.} to \lipItem{4.} described above, at most one star with center $c(e_i)$ with $e_i \in E(A,B)$ is included in $H'$.
Each such star is associated with the choice of an edge of $G$.
We build a perfect and consistent matching $M$ by taking all these edges (notice that exactly $|V_1|$ are taken), and the result follows.
\end{proof}

\medskip
\noindent\textbf{Disjoint union of $2$-out-stars ($2$-in-stars) plus a $k$-homogeneous star.}
For this case, it suffices to repeat the proof of the previous case and, in the end, add one homogeneous star with center $u$ and having as leaves all vertices of the form $c(e)$ added to $D$ in the beginning of the construction, plus all vertices of $F$, and set $k = |V_1| + k'$, where $V_1$ and $k'$ are as defined in the aforementioned proof.
Since $G$ has maximum degree three, this ensures that the only possible choice for the center of start with degree $k$ in $H'$ is $u$.

\medskip
\noindent\textbf{Subdivisions of homogeneous stars.}
For this case, we need a special version of $3$-\textsc{SAT}.
The $3$-\textsc{SAT}-$(2,2)$ problem is as the classical $3$-\textsc{SAT} but with the additional restriction that each literal appears \emph{exactly} twice in the clauses.
In other words, for each variable $x$ the literal $x$ appears twice and the literal $\overline{x}$ appears twice.
This problem was shown to be \NP-complete (and even hard to approximate) by Berman et al.~\cite{Berman2003}.

\begin{proof}[Proof of item \lipItem{5.} of \autoref{theorem:hardness results}.]
We first do the proof for the case when both $H_1$ and $H_2$ are built from subdivisions of out-stars, and then argue that the other cases follow as well.
Let $\mathcal{I}$ be an instance of $3$-\textsc{SAT}-$(2,2)$ with variables $x_1, \ldots, x_n$ and clauses $C_1, \ldots, C_m$.
From $\mathcal{I}$ we build a DAG $D$ as follows.
Let $S$ be the digraph formed by $1$-subdividing all edges of a $2$-out-star.
For each $x_i$ with $i \in [n]$, add to $D$ a copy of the digraph $S$ with center $v_i$.
Name the vertices of the paths leaving $v_i$ as $x^i_1, x^i_2$ and $\overline{x}^i_1, \overline{x}^i_2$, respectively.
Then, add to $D$ the \emph{variable selector} vertex $s$ and an arc from $s$ to the center $v_i$ of every copy of $S$ added to $D$.

Now, add to $D$ a \emph{clause verifier} vertex $c$ and, for each clause $C_i$ with $i \in [m]$, add to $D$ a vertex $y_i$ and an arc from $c$ to $y_i$.
Then, for each $i \in [n]$,
\begin{enumerate}
    \item an arc from $y_j$ to $x^i_1$ and an arc from $y_\ell$ to $x^i_2$, where $C_j$ and $C_\ell$ are the two clauses containing the literal $x_i$; and
    \item an arc from $y_j$ to $\overline{x}^i_1$ and an arc from $y_\ell$ to $\overline{x}^i_2$, where $C_j$ and $C_\ell$ are the two clauses containing the literal $\overline{x_i}$.
\end{enumerate}

See \autoref{fig:construction-for-two-long-stars} for an example of this construction.

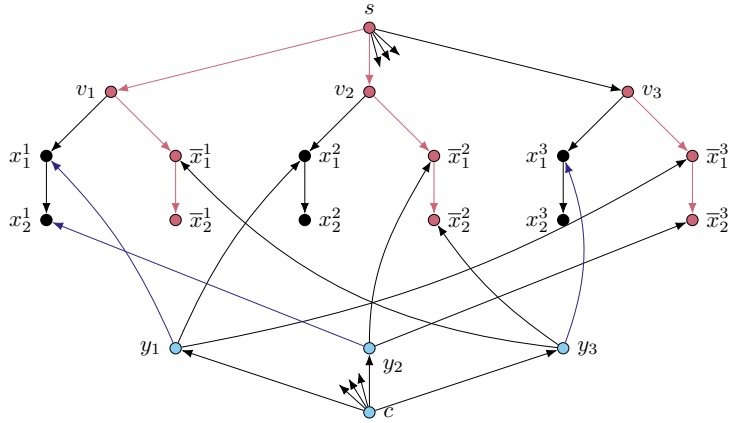
\begin{figure}[h]
\centering
\scalebox{.85}{\begin{tikzpicture}

\node[redvertex, label=90:{$s$}] (s) at (4,2) { };

\node (f1) at ($(s) + (1/5, -.75)$) { };
\node (f2) at ($(s) + (2/5, -.65)$) { };
\node (f3) at ($(s) + (3/5, -.55)$) { };

\foreach \i in {1,...,3} {
    \draw[arrow] (s) -- (f\i);
}

\begin{scope}[xscale=1, yshift=1cm]
\node[redvertex, label=180:$v_1$] (v1) at (0,0) {};
\node[blackvertex, label=180:{$x^1_1$}] (x11) at ($(v1) + (-1,-1)$) {};
\node[blackvertex, label=180:{$x^1_2$}] (x12) at ($(x11) + (0,-1)$) {};
\node[redvertex, label=0:{$\overline{x}^1_1$}] (nx11) at ($(x11) + (2,0)$) {};
\node[redvertex, label=0:{$\overline{x}^1_2$}] (nx12) at ($(x12) + (2,0)$) {};
\draw[arrow] (v1) to (x11);
\draw[arrow, goodred] (v1) to (nx11);
\draw[arrow] (x11) to (x12);
\draw[arrow, goodred] (nx11) to (nx12);
\draw[arrow, goodred] (s) to (v1);
\end{scope}

\begin{scope}[xshift=4cm,xscale=1,yshift=1cm]
\node[redvertex, label=180:$v_2$] (v2) at (0,0) {};
\node[blackvertex, label=0:{$x^2_1$}] (x21) at ($(v2) + (-1,-1)$) {};
\node[blackvertex, label=0:{$x^2_2$}] (x22) at ($(x21) + (0,-1)$) {};
\node[redvertex, label=0:{$\overline{x}^2_1$}] (nx21) at ($(x21) + (2,0)$) {};
\node[redvertex, label=0:{$\overline{x}^2_2$}] (nx22) at ($(x22) + (2,0)$) {};
\draw[arrow] (v2) to (x21);
\draw[arrow, goodred] (v2) to (nx21);
\draw[arrow] (x21) to (x22);
\draw[arrow, goodred] (nx21) to (nx22);
\draw[arrow, goodred] (s) to (v2);
\end{scope}

\begin{scope}[xshift=8cm,xscale=1,yshift=1cm]
\node[redvertex, label=0:$v_3$] (v3) at (0,0) {};
\node[blackvertex, label=180:{$x^3_1$}] (x31) at ($(v3) + (-1,-1)$) {};
\node[blackvertex, label=180:{$x^3_2$}] (x32) at ($(x31) + (0,-1)$) {};
\node[redvertex, label=0:{$\overline{x}^3_1$}] (nx31) at ($(x31) + (2,0)$) {};
\node[redvertex, label=0:{$\overline{x}^3_2$}] (nx32) at ($(x32) + (2,0)$) {};
\draw[arrow] (v3) to (x31);
\draw[arrow, goodred] (v3) to (nx31);
\draw[arrow] (x31) to (x32);
\draw[arrow, goodred] (nx31) to (nx32);
\draw[arrow] (s) to (v3);
\end{scope}

%\begin{scope}[yshift = 1cm]
%\node[blackvertex, label=180:$x_1$] (h1) at (0, -4) {};
%\node[blackvertex, label=180:$x_2$] (h2) at ($(h1) + (1.5,0)$) {};
%\node[blackvertex, label=180:$\overline{x_3}$] (h3) at ($(h2) + (1.5,0)$) {};

%\node[blackvertex, label=180:$\overline{x_3}$] (h6) at (8, -4) {};
%\node[blackvertex, label=180:$\overline{x_2}$] (h5) at ($(h6) + (-1.5,0)$) {};
%\node[blackvertex, label=180:$x_1$] (h4) at ($(h5) + (-1.5,0)$) {};
%\end{scope}

\node[bluevertex, label=0:$c$] (c) at (4, -4) {};
\node[bluevertex, label=180:$y_1$] (y1) at ($(c) + (-3,1)$) {};
\node[bluevertex, label=-30:$y_2$] (y2) at ($(c) + (0,1)$) {};
\node[bluevertex, label=0:$y_3$] (y3) at ($(c) + (3,1)$) {};
%\node[blackvertex, label=90:$c$] (c) at (4, -4) {};
%\foreach \i in {1,2,3}
%	\draw[arrow] (y1) to (h\i);

%\foreach \i in {4,5,6}
%	\draw[arrow] (y2) to (h\i);

\draw[arrow] (c) to (y1);
\draw[arrow] (c) to (y2);
\draw[arrow] (c) to (y3);
\draw[arrow, goodblue] (y1) to [bend right = 10] (x11);
\draw[arrow] (y1) to [bend left = 10] (x21);
\draw[arrow] (y1) to [bend right = 10] (nx31);
\draw[arrow, goodblue] (y2) to (x12);
\draw[arrow] (y2) to [out=90, in = -120] (nx21);
\draw[arrow] (y2) to (nx32);
\draw[arrow] (y3) to [bend left = 20] (nx11);
\draw[arrow] (y3) to [bend left = 10] (nx22);
\draw[arrow, goodblue] (y3) to [bend right = 20] (x31);

\node (f1) at ($(c) + (-1/5, .75)$) { };
\node (f2) at ($(c) + (-2/5, .65)$) { };
\node (f3) at ($(c) + (-3/5, .55)$) { };

\foreach \i in {1,...,3} {
    \draw[arrow] (c) -- (f\i);
}
\end{tikzpicture}}%
\caption{Example of the construction from the proof of item \lipItem{5.} of \autoref{theorem:hardness results}. In the example, the vertex $y_1$ is associated with the clause $(x_1, x_2, \overline{x_3})$, vertex $y_2$ with the clause $(x_1, \overline{x_2}, \overline{x_3})$, and $y_3$ with the clause $(\overline{x_1}, \overline{x_2}, x_3)$. Possible choices for the out-stars with root $s$ and $c$ are in red and blue, respectively.}
\label{fig:construction-for-two-long-stars}
\end{figure}

Let $S_1$ be the digraph formed by $2$-subdividing every arc of an $n$-out-star, $S_2$ be the digraph formed by $1$-subdividing every arc of an $m$-out-star, and $H$ be formed by the disjoint union of $S_1$ and $S_2$.
We claim that $\mathcal{I}$ is a positive instance if and only if $D$ contains a subdigraph $H' \cong H$.

First, assume that $\mathcal{I}$ is positive.
For the $n$-out-star, it suffices to take, for each $i \in [n]$, the path from $s$ to $v_i$ and then to $x^i_2$ if the solution to $\mathcal{I}$ chooses $x_i =$ {\sf false}, or the path from $s$ to $v_i$ and then to $\overline{x}^i_2$ otherwise.
For the $m$-out-star, we proceed as follows.
Start by taking $c$ for the center and then the path from $c$ to each vertex $y_i$ with $i \in [m]$.
For each clause $C_i$, choose a literal $\ell_i$ of the clause evaluating to {\sf true}.
This implies that $x^j_1, x^j_2$ are not in the $n$-out-star if $\ell_i$ is an occurrence of $x_j$, and that $\overline{x}^j_1, \overline{x}^j_2$ are not in the $n$-out-star if $\ell_i$ is an occurrence of $\overline{x_j}$.
Thus, for each $i \in [m]$, $y_i$ has an out-neighbor not used by the $n$-out-star.
For each $y_i$ we take one such out-neighbor $w$ and add it, together with the arc from $y_i$ to $w$, to the $m$-out-star and the necessity follows.

Assume now that there is $H' \cong H$ with $H' \subseteq D$.
There is only one vertex of out-degree $n$ and one vertex of out-degree $m$ in $D$: $s$ and $c$, respectively.
Thus the centers of the stars $S'_1$ and $S'_2$ forming $H'$ are $s$ and $c$, respectively.
By construction, for each $i \in [n]$ the out-star $S'_1$ contains a path from $s$ to $v_i$ then to exactly one of the two paths leaving $v_i$.
Thus we set $x_i =$ {\sf false} if the path ending in $x^i_2$ is in $S'_1$ and $x_i =$ {\sf true} otherwise.
Now for each $v_i$ at most one of the paths leaving $v_i$ contains a vertex that is a leaf of $S'_2$.
Since $m$ paths are leaving $c$ in $S'_2$, each containing the vertex $y_i$ associated with the clause $C_i$, this implies that every  clause $C_i$ contains a literal associated with a vertex of a path leaving $s$ not taken by $S'_1$.
In other words, each clause contains a literal that evaluates to {\sf true} and the result follows.

Clearly $D$ is a DAG since vertices $s$ and $c$ are sources and all the maximal paths reaching vertices which are not sources start in either $s$ and $c$.

To see that the result also holds for an in-star and an out-star, it suffices to reverse the directions of the paths leaving $s$ or the directions of the paths leaving $c$.
If the directions of all such paths are reversed, the case of two in-stars is covered.
\end{proof}

\medskip
\noindent\textbf{Proof for caterpillars.}
We can obtain a proof for item \lipItem{6.} of \autoref{theorem:hardness results} by adapting the proof for item \lipItem{3.}
The reduction is from the same matching problem in bipartite graphs (see~\autoref{proposition:np-complete-matching-bipartite} for the precise statement of this problem).
For clarity, we repeat the first steps of the construction.
Given a bipartite graph $G$ with partition $\{V_1, V_2\}$ of $V(G)$, we first orient every edge from $V_1$ to $V_2$ to build a DAG $D$.
Then for each $e \in E(G)$ with endpoints $u$ and $v$, we add to $D$ a vertex $c(e)$ with out-neighbors $u$ and $v$.
The proof of item \lipItem{3.} of \autoref{theorem:hardness results} then adds another set of $2$-out-stars to $D$ which are used to ensure that, if the stars are found, then a matching is found in $G$ with the desired properties.
Namely, we show that $G$ contains a perfect matching that is consistent with the given partitions $\mathcal{P}_1$ and $\mathcal{P}_2$ of $V_1$ and $V_2$, respectively, if and only if $|V_1| + k'$ pairwise disjoint $2$-out-stars are found in the constructed digraph $D$, for some choice of $k'$ depending on the configurations of the edges of $G$ between parts of $\mathcal{P}_1$ and $\mathcal{P}_2$.

To adapt the proof for the caterpillar case, we focus on the $2$-out-stars with centers $c(e)$, for every $e \in E(G)$.
Let $m = |E(G)|$.
Notice that we can assume that $m > |V_1|$ otherwise the matching problem is trivial.
We also remind the reader that we can assume that $G$ has maximum degree three.
The construction discussed in the previous paragraph added to $D$ exactly $m$ vertices of the form $c(e)$ with $e\in E(G)$.
Since in $G$ we are searching for a perfect matching, exactly $|V_1|$ of those stars need to be selected.
The goal is to define the target digraph $H$ as a caterpillar ending on an out-star with large degree, and use the large star of the caterpillar to select which vertices of the form $c(e)$ are used as centers of $2$-out-stars in the caterpillar.
The remainder of the construction, which guarantees the consistency of the matching, follows exactly as in the proof of \lipItem{3.} of \autoref{theorem:hardness results}.

Formally, let $C = \{c(e) \in V(D) \mid e \in E(G)\}$ and consider an arbitrary order $c(e_1), \ldots, c(e_m)$ of $C$.
For each $i \in [m]$, we add an arc from $c(e_i)$ to $c(e_j)$ for all $j > i$.
In other words, now $C$ induces a transitive tournament in $D$.
Now, we subdivide once every arc of $D[C]$ and let $U$ be the set of all vertices added to $D$ by the $1$-subdivisions.
Then, we add another vertex $r$ to $D$ and an arc from each vertex in $C$ to $r$, and an arc from $r$ to each vertex in $U$.
We define $H$ to be the caterpillar on a directed path with $\ell = 2|V_1|+1$ vertices $\{x_1, \ldots, x_\ell, r'\}$, enumerated as they appear in the spine, where the last vertex $r'$ is the center of an out-star with $\binom{m}{2}-|V_1|$ leaves, and every $x_i$ with odd $i \in [\ell]$ a branching vertex having exactly two leaves as out-neighbors. 
From this choice of $H$, the remainder of the proof follows as in the proof of item \lipItem{3.} of \autoref{theorem:hardness results}.
The key observation is that there is only one choice in $D$ for the center $r'$ of the largest star of the caterpillar: the vertex $r$.
This implies that $\binom{m}{2}-|V_1|$ vertices of $U$ are taken as leaves of $r$ and thus exactly $|V_1|$ can be used to build the spine path of the caterpillar.
Informally, the vertices of $U$ used by the spine of the caterpillar are used to identify the $|V_1|$ vertices of the form $c(e)$ that are used to select edges of $G$.
Notice that with this construction, although $D$ is no longer a DAG, its directed treewidth is one since all cycles use the vertex $r$.

\section{Conclusions and further research}
\label{sec:conclusions}

Our article initiates a systematic study of the following fascinating question. Let  ${\cal A}$ be a collection of {\sl allowed} digraphs, let $H$ be a digraph that can be built as the union of at most $k$ digraphs in ${\cal A}$, and let $D$ be a digraph with directed treewidth at most $w$. For which collections ${\cal A}$ the problem of deciding whether $D$ contains a subdigraph isomorphic to $H$ can be solved in \XP time parameterized by $k$ and $w$? Let us call a collection ${\cal A}$ \textit{easy} (resp. \textit{hard}) if the answer to the previous question is positive (resp. negative). Note that a negative answer will be typically conditioned to some complexity hypothesis, such as ${\sf P} \neq \NP$ or another ones. 

Our \XP algorithm (cf. \autoref{thm:our-XP-algo}) shows that if ${\cal A}$ contains the directed paths and the stars (oriented away from or towards the center), then it is easy. On the other hand, our hardness results (cf. \autoref{theorem:hardness results}) provide a number of examples of hard collections ${\cal A}$ containing digraphs that are ``close'' to paths and stars (cf. \autoref{fig:hardness-digraphs-examples}), namely by showing that the problem is \NP-complete for fixed small values of $k$ and $w$. 

The main message of our article is that a good collection ${\cal A}$ is unlikely to contain many digraphs that differ substantially from (combinations of) paths and stars. Nevertheless, we fall short of providing a full characterization of easy and hard collections. This seems to be a very challenging question, in particular because of the reasons that we proceed to discuss. 

For a positive integer $p$, we say that a collection of digraphs ${\cal A}$ is $p$-\textit{easy} (resp. $p$-\textit{hard}) if the answer to the above question is positive (resp. negative) for every $k \leq p$ (resp. for some $k \leq p$ ). Hence, a collection ${\cal A}$ is easy if and only if it is $p$-easy for every $p \geq 1$. Let $H_{\ell}$ be the digraph made of $\ell$ pairwise disjoint arcs. Then deciding whether a host digraph $D$ contains $H_{\ell}$ as a subdigraph can be done in polynomial time for any $\ell$, by just running a classical maximum matching algorithm~\cite{Graph.Theory} in the underlying undirected graph of $D$. Thus, the collection ${\cal A}= \{H_{\ell} \mid \ell \in \mathds{N}\}$ is $1$-easy, but its elements {\sl cannot} be formed as the union of few paths or stars, and therefore it escapes the setting of \autoref{thm:our-XP-algo}. Note, however, that this collection ${\cal A}$ is 2-hard, because an antidirected path (cf. \autoref{fig:hardness-digraphs-examples}\textbf{(c)}) can be obtained as the union of two elements in ${\cal A}$, and we have proved that the corresponding problem is \NP-complete in DAGs (item \lipItem{1.} of \autoref{theorem:hardness results}). 

Exploiting the potential of a maximum matching algorithm, we can construct other $1$-easy collections ${\cal A}'$ that escape the setting of \autoref{thm:our-XP-algo}. Indeed, let $H_{\ell}'$ be the digraph obtained from an out-star with $\ell$ leaves by subdividing every arc once (and keeping the same orientation in both new arcs), and let ${\cal A}'= \{H_{\ell}' \mid \ell \in \mathds{N}\}$. Note that the elements in ${\cal A}'$ cannot be obtained as the union of few paths or stars, and therefore it also escapes the setting of \autoref{thm:our-XP-algo}. We claim that ${\cal A}'$ is $1$-easy. Indeed, we proceed to present a polynomial-time algorithm  to decide whether a digraph $D$ contains a subdigraph isomorphic to $H_{\ell}'$, for any $\ell \geq 0$ (without even needing to use a bound on the directed treewidth of $D$). For a vertex $v \in V(D)$, let $N_v$ be the out-neighborhood of $v$, let $A_v$ be the set of arcs of $D$ having a vertex in $N_v$ as the tail (note that the head may also be in $N_v$), and let $D_v$ be the subdigraph of $D$ defined by the union of the arcs in $A_v$. It is easy to see that $D$ contains a subdigraph isomorphic to $H_{\ell}'$ if and only if there is a vertex $v \in V(D)$ such that the underlying graph of $D_v$ contains a matching of size at least $\ell$, which can be clearly decided in polynomial time. A symmetric argument applies to the collection containing subdivisions of in-stars. 

How far can we go with this matching-based approach? A next natural step could be, for instance, try to generalize the above algorithm to find a digraph $H$ defined by the disjoint union of {\sl two} subdivided stars $H_{\ell_1}'$ and $H_{\ell_2}'$, as defined in the previous paragraph, for two  positive integers $\ell_1$ and $\ell_2$. Note that this digraph $H$ is really close to the one depicted in \autoref{fig:hardness-digraphs-examples}\textbf{(e)} (hence, a hard collection), also consisting of two subdivided out-stars, but one subdivided once and the other one subdivided {\sl twice}, while in $H$ both stars are subdivided {\sl once}. It turns out that finding $H$ has a strong connection with a notoriously open problem called {\sc Exact Matching}. In this problem, we are given an edge colored undirected graph $G$, with colors red and blue, and an integer $\ell$, the goal is to decide whether or not the graph contains a perfect matching with exactly $\ell$ red edges. {\sc Exact Matching} is known to be solvable in randomized polynomial time~\cite{MulmuleyVV87} (hence, in the class {\sf RP} and unlikely to be \NP-hard), but it is still unknown whether it is in \P (that is, solvable in {\sl deterministic} polynomial time). In particular, {\sc Exact Matching} is known to be in {\sf P} only for very restricted graph classes, such as cliques, complete bipartite graphs, graphs of bounded independence number, or planar graphs; see~\cite{Maalouly24,MaaloulyHW24,Maalouly23} for the recent work of El Maalouly on this problem.

Back to our problem, consider a digraph $D$ with only two vertices $v_1$ and $v_2$ of ``large'' out-degree, and suppose that $|V(D)|$ is even. Let $\ell_1$ and $\ell_2$ be positive integers such that $\ell_1 + \ell_2 = (|V(D)| -2) / 2$, and suppose that both $\ell_1$ and $\ell_2$ are ``large'', in the sense that there is no vertex in $V(D) \setminus \{v_1,v_2\}$ with out-degree at least $\min\{\ell_1,\ell_2\}$. We construct from $D$ an undirected graph $G$, with edges colored red and blue, as follows. We define $G$ to be the underlying graph of $D \setminus \{v_1,v_2\}$, and we color an edge of $G$ red (resp. blue) if at least one of its endpoints is an out-neighbor of $v_1$ (resp. $v_2$) in $D$; we can add parallel edges if we want to avoid the same edge to be colored red and blue. It can be easily verified that $D$ contains the disjoint union of $H_{\ell_1}'$ and $H_{\ell_2}'$ as a subdigraph if and only if $G$ admits a perfect matching using exactly $\ell_1$ red edges. We leave this elusive case as an open problem.

In order to unveil other easy collections, it would be desirable to integrate the matching ``technology'' into the techniques that we used to obtain our \XP algorithm (cf. \autoref{sec:algorithm}), namely the approach based on solving an integer system of linear inequations to deal with the stars, and the approach based on dynamic programming on an arboreal decomposition to deal with the paths. Unfortunately, this integration seems quite challenging, since matching algorithms have a {\sl global} nature, but both approaches in our \XP algorithm have a {\sl local} nature (in the case of stars it is clear, and in the case of paths, the local nature is given by the few crossings that the paths can have in the dynamic programming algorithm, as formalized in \autoref{proposition:paths-breakability}). Thus, we think that for finding other easy collections ${\cal A}$, substantially new algorithmic ideas are needed. 

Finally, it would be very interesting to explore the existence of a generalization of the logical meta-theorem of de Oliveira Oliveira~\cite{Oliveira16} to subgraphs $H$ that can be formed by the union of $k$ directed paths {\sl or stars}. 

%%%% END OF TEXT %%%%
\bibliography{main}

\begin{thebibliography}{10}

\bibitem{doi:10.1137/0608024}
Stefan Arnborg, Derek~G. Corneil, and Andrzej. Proskurowski.
\newblock Complexity of finding embeddings in a $k$-tree.
\newblock {\em SIAM Journal on Algebraic Discrete Methods}, 8(2):277--284,
  1987.
\newblock \href {https://doi.org/10.1137/0608024} {\path{doi:10.1137/0608024}}.

\bibitem{Classes.Directed.Graphs}
J.~Bang-Jensen and G.~Gregory.
\newblock {\em Classes of Directed Graphs}.
\newblock Springer Monographs in Mathematics, 2018.
\newblock \href {https://doi.org/10.1007/978-3-319-71840-8}
  {\path{doi:10.1007/978-3-319-71840-8}}.

\bibitem{BANGJENSEN201768}
Jørgen Bang-Jensen, Stéphane Bessy, Bill Jackson, and Matthias Kriesell.
\newblock Antistrong digraphs.
\newblock {\em Journal of Combinatorial Theory, Series B}, 122:68--90, 2017.
\newblock \href {https://doi.org/10.1016/j.jctb.2016.05.004}
  {\path{doi:10.1016/j.jctb.2016.05.004}}.

\bibitem{Berman2003}
Piotr Berman, Marek Karpinski, and Alex~D. Scott.
\newblock Approximation hardness and satisfiability of bounded occurrence
  instances of {SAT}.
\newblock {\em Electronic Colloquium on Computational Complexity (ECCC)},
  {TR03-022}, 2003.
\newblock URL:
  \url{https://eccc.weizmann.ac.il/eccc-reports/2003/TR03-022/index.html}.

\bibitem{Bodlaender96}
Hans~L. Bodlaender.
\newblock A linear-time algorithm for finding tree-decompositions of small
  treewidth.
\newblock {\em {SIAM} Journal on Computing}, 25(6):1305--1317, 1996.
\newblock \href {https://doi.org/10.1137/S0097539793251219}
  {\path{doi:10.1137/S0097539793251219}}.

\bibitem{Graph.Theory}
A.~Bondy and M.~Ram Murty.
\newblock {\em Graph Theory}.
\newblock Springer-Verlag London, 2008.

\bibitem{doi:10.1137/21M1452664}
Victor Campos, Raul Lopes, Ana~Karolinna Maia, and Ignasi Sau.
\newblock Adapting the directed grid theorem into an fpt algorithm.
\newblock {\em SIAM Journal on Discrete Mathematics}, 36(3):1887--1917, 2022.
\newblock \href {https://doi.org/10.1137/21M1452664}
  {\path{doi:10.1137/21M1452664}}.

\bibitem{COURCELLE199012}
Bruno Courcelle.
\newblock {The monadic second-order logic of graphs. I. Recognizable sets of
  finite graphs}.
\newblock {\em Information and Computation}, 85(1):12--75, 1990.
\newblock \href {https://doi.org/10.1016/0890-5401(90)90043-H}
  {\path{doi:10.1016/0890-5401(90)90043-H}}.

\bibitem{CyganFKLMPPS15}
Marek Cygan, Fedor~V. Fomin, Lukasz Kowalik, Daniel Lokshtanov, D{\'{a}}niel
  Marx, Marcin Pilipczuk, Michal Pilipczuk, and Saket Saurabh.
\newblock {\em Parameterized Algorithms}.
\newblock Springer, 2015.

\bibitem{Oliveira16}
Mateus de~Oliveira~Oliveira.
\newblock An algorithmic metatheorem for directed treewidth.
\newblock {\em Discrete Applied Mathematics}, 204:49--76, 2016.
\newblock \href {https://doi.org/10.1016/j.dam.2015.10.020}
  {\path{doi:10.1016/j.dam.2015.10.020}}.

\bibitem{DF13}
R.~G. Downey and M.~R. Fellows.
\newblock {\em Fundamentals of Parameterized Complexity}.
\newblock Texts in Computer Science. Springer, 2013.

\bibitem{GanianHK0ORS16}
Robert Ganian, Petr Hlinen{\'{y}}, Joachim Kneis, Daniel Meister, Jan
  Obdrz{\'{a}}lek, Peter Rossmanith, and Somnath Sikdar.
\newblock Are there any good digraph width measures?
\newblock {\em Journal of Combinatorial Theory, Series {B}}, 116:250--286,
  2016.
\newblock \href {https://doi.org/10.1016/J.JCTB.2015.09.001}
  {\path{doi:10.1016/J.JCTB.2015.09.001}}.

\bibitem{GiannopoulouKKK20}
Archontia~C. Giannopoulou, Ken{-}ichi Kawarabayashi, Stephan Kreutzer, and
  O{-}joung Kwon.
\newblock {The Directed Flat Wall Theorem}.
\newblock In {\em Proc. of the 13st {ACM-SIAM} Symposium on Discrete Algorithms
  (SODA)}, pages 239--258, 2020.
\newblock \href {https://doi.org/10.1137/1.9781611975994.15}
  {\path{doi:10.1137/1.9781611975994.15}}.

\bibitem{HatzelKMM24}
Meike Hatzel, Stephan Kreutzer, Marcelo~Garlet Milani, and Irene Muzi.
\newblock Cycles of well-linked sets and an elementary bound for the directed
  grid theorem.
\newblock In {\em Proc. of the 65th {IEEE} Annual Symposium on Foundations of
  Computer Science (FOCS)}, pages 1--20, 2024.
\newblock \href {https://doi.org/10.1109/FOCS61266.2024.00011}
  {\path{doi:10.1109/FOCS61266.2024.00011}}.

\bibitem{Johnson.Robertson.Seymour.Thomas.01}
T.~Johnson, N.~Robertson, P.~D. Seymour, and R.~Thomas.
\newblock Directed tree-width.
\newblock {\em Journal of Combinatorial Theory, Series B}, 82(01):138--154,
  2001.
\newblock \href {https://doi.org/10.1006/jctb.2000.2031}
  {\path{doi:10.1006/jctb.2000.2031}}.

\bibitem{Lenstra83}
Hendrik W.~Lenstra Jr.
\newblock Integer programming with a fixed number of variables.
\newblock {\em Mathematics of Operations Research}, 8(4):538--548, 1983.
\newblock \href {https://doi.org/10.1287/MOOR.8.4.538}
  {\path{doi:10.1287/MOOR.8.4.538}}.

\bibitem{KawarabayashiK15}
Ken{-}ichi Kawarabayashi and Stephan Kreutzer.
\newblock {Towards the Graph Minor Theorems for Directed Graphs}.
\newblock In {\em Proc. of the 42nd International Colloquium on Automata,
  Languages, and Programming (ICALP)}, volume 9135 of {\em LNCS}, pages 3--10,
  2015.
\newblock \href {https://doi.org/10.1007/978-3-662-47666-6\_1}
  {\path{doi:10.1007/978-3-662-47666-6\_1}}.

\bibitem{KINTALI201783}
Shiva Kintali.
\newblock Directed width parameters and circumference of digraphs.
\newblock {\em Theoretical Computer Science}, 659:83--87, 2017.
\newblock \href {https://doi.org/10.1016/j.tcs.2016.10.010}
  {\path{doi:10.1016/j.tcs.2016.10.010}}.

\bibitem{KorhonenL23}
Tuukka Korhonen and Daniel Lokshtanov.
\newblock An improved parameterized algorithm for treewidth.
\newblock In {\em Proc. of the 55th Annual {ACM} Symposium on Theory of
  Computing (STOC)}, pages 528--541, 2023.
\newblock \href {https://doi.org/10.1145/3564246.3585245}
  {\path{doi:10.1145/3564246.3585245}}.

\bibitem{LampisKM11}
Michael Lampis, Georgia Kaouri, and Valia Mitsou.
\newblock On the algorithmic effectiveness of digraph decompositions and
  complexity measures.
\newblock {\em Discrete Optimization}, 8(1):129--138, 2011.
\newblock \href {https://doi.org/10.1016/J.DISOPT.2010.03.010}
  {\path{doi:10.1016/J.DISOPT.2010.03.010}}.

\bibitem{LopesS22}
Raul Lopes and Ignasi Sau.
\newblock A relaxation of the directed disjoint paths problem: {A} global
  congestion metric helps.
\newblock {\em Theoretical Computer Science}, 898:75--91, 2022.
\newblock \href {https://doi.org/10.1016/J.TCS.2021.10.023}
  {\path{doi:10.1016/J.TCS.2021.10.023}}.

\bibitem{Maalouly23}
Nicolas~El Maalouly.
\newblock {Exact Matching: Algorithms and Related Problems}.
\newblock In {\em Proc. of the 40th International Symposium on Theoretical
  Aspects of Computer Science (STACS)}, volume 254 of {\em LIPIcs}, pages
  29:1--29:17, 2023.
\newblock \href {https://doi.org/10.4230/LIPICS.STACS.2023.29}
  {\path{doi:10.4230/LIPICS.STACS.2023.29}}.

\bibitem{Maalouly24}
Nicolas~El Maalouly.
\newblock {\em {Towards a Deterministic Polynomial Time Algorithm for the Exact
  Matching Problem}}.
\newblock PhD thesis, {ETH} Zurich, Z{\"{u}}rich, Switzerland, 2024.
\newblock \href {https://doi.org/10.3929/ETHZ-B-000675444}
  {\path{doi:10.3929/ETHZ-B-000675444}}.

\bibitem{MaaloulyHW24}
Nicolas~El Maalouly, Sebastian Haslebacher, and Lasse Wulf.
\newblock On the exact matching problem in dense graphs.
\newblock In {\em Proc. of the 41st International Symposium on Theoretical
  Aspects of Computer Science (STACS)}, volume 289 of {\em LIPIcs}, pages
  33:1--33:17, 2024.
\newblock \href {https://doi.org/10.4230/LIPICS.STACS.2024.33}
  {\path{doi:10.4230/LIPICS.STACS.2024.33}}.

\bibitem{MarxP14}
D{\'{a}}niel Marx and Michal Pilipczuk.
\newblock {Everything you always wanted to know about the parameterized
  complexity of Subgraph Isomorphism (but were afraid to ask)}.
\newblock In {\em Proc. of the 31st International Symposium on Theoretical
  Aspects of Computer Science (STACS)}, volume~25 of {\em LIPIcs}, pages
  542--553, 2014.
\newblock \href {https://doi.org/10.4230/LIPICS.STACS.2014.542}
  {\path{doi:10.4230/LIPICS.STACS.2014.542}}.

\bibitem{MulmuleyVV87}
Ketan Mulmuley, Umesh~V. Vazirani, and Vijay~V. Vazirani.
\newblock Matching is as easy as matrix inversion.
\newblock {\em Combinatorica}, 7(1):105--113, 1987.
\newblock \href {https://doi.org/10.1007/BF02579206}
  {\path{doi:10.1007/BF02579206}}.

\bibitem{PLAISTED198065}
David~A. Plaisted and Shmuel Zaks.
\newblock {An NP-complete matching problem}.
\newblock {\em Discrete Applied Mathematics}, 2(1):65--72, 1980.
\newblock \href {https://doi.org/10.1016/0166-218X(80)90055-4}
  {\path{doi:10.1016/0166-218X(80)90055-4}}.

\bibitem{Reed99}
Bruce~A. Reed.
\newblock Introducing directed tree width.
\newblock {\em Electronic Notes in Discrete Mathematics}, 3:222--229, 1999.
\newblock \href {https://doi.org/10.1016/S1571-0653(05)80061-7}
  {\path{doi:10.1016/S1571-0653(05)80061-7}}.

\bibitem{RobertsonS90a}
Neil Robertson and Paul~D. Seymour.
\newblock Graph minors. {IV.} tree-width and well-quasi-ordering.
\newblock {\em Journal of Combinatorial Theory, Series {B}}, 48(2):227--254,
  1990.
\newblock \href {https://doi.org/10.1016/0095-8956(90)90120-O}
  {\path{doi:10.1016/0095-8956(90)90120-O}}.

\bibitem{Slivkins.03}
Aleksandrs Slivkins.
\newblock Parameterized tractability of edge-disjoint paths on directed acyclic
  graphs.
\newblock {\em SIAM Journal on Discrete Mathematics}, 24(1):146--157, 2010.
\newblock \href {https://doi.org/10.1137/070697781}
  {\path{doi:10.1137/070697781}}.

\end{thebibliography}
\end{document}